\documentclass[]{article}

\usepackage[utf8]{inputenc} 
\usepackage[T1]{fontenc}    
\usepackage{hyperref}       
\usepackage{url}            
\usepackage{booktabs}       
\usepackage{amsmath}
\usepackage{amsfonts}       
\usepackage{nicefrac}       
\usepackage{microtype}      
\usepackage{xcolor}         
\usepackage{amsthm}
\usepackage{bbm}
\usepackage{graphicx}
\usepackage[ruled, vlined,algo2e]{algorithm2e}
\usepackage{cleveref}
\usepackage{svg}

\usepackage{caption}
\usepackage{subcaption}

\usepackage[accepted]{icml2023}

\icmltitlerunning{Optimally-weighted Estimators of the Maximum Mean Discrepancy for Likelihood-Free Inference}

\newtheorem{theorem}{Theorem}

\newtheorem{lemma}{Lemma}

\newtheorem{definition}{Definition}

\newtheorem{assumption}{Assumption}

\crefname{assumption}{}{}

\newtheorem{oldassumption}{Assumption}

\crefname{oldassumption}{}{}

\def\X{\mathcal{X}}

\def\R{\mathbb{R}}
\def\Q{\mathbb{Q}}
\def\P{\mathbb{P}}
\def\PX{\mathcal{P}(\X)}
\def\Pk{\mathcal{P}_k(\X)}
\def\U{\mathbb{U}}
\def\W{\mathcal{W}}
\def\L{\mathcal{L}}

\def\d{\text{d}}

\def\calU{\mathcal{U}}
\def\Hk{\mathcal{H}_k}
\def\Hc{\mathcal{H}_c}

\def\Hct{\mathcal{H}_{c_\theta}}

\newenvironment{talign*}
 {\csname align*\endcsname}
 {\endalign}
\newenvironment{talign}
{\align}
{\endalign}

\newcommand{\fxb}[1]{}

\begin{document}


\twocolumn[

\icmltitle{Optimally-weighted Estimators of the Maximum Mean Discrepancy for Likelihood-Free Inference}


\begin{icmlauthorlist}
    \icmlauthor{Ayush Bharti}{aalto}
    \icmlauthor{Masha Naslidnyk}{ucl}
    \icmlauthor{Oscar Key}{ucl}
    \icmlauthor{Samuel Kaski}{aalto,manchester}
    \icmlauthor{Fran\c{c}ois-Xavier Briol}{ucl}
\end{icmlauthorlist}

\icmlaffiliation{aalto}{Department of Computer Science, Aalto University, Espoo, Finland}
\icmlaffiliation{manchester}{Department of Computer Science, University of Manchester, Manchester, United Kingdom}
\icmlaffiliation{ucl}{Department of Statistical Science, University College London, London, United Kingdom}

\icmlcorrespondingauthor{Ayush Bharti}{ayush.bharti@aalto.fi}

\icmlkeywords{Machine Learning, ICML, Likelihood-free inference, Maximum mean discrepancy}
\vskip 0.3in
]



\printAffiliationsAndNotice{}  

\begin{abstract}
  Likelihood-free inference methods typically make use of a distance between simulated and real data. A common example is the maximum mean discrepancy (MMD), which has previously been used for approximate Bayesian computation, minimum distance estimation, generalised Bayesian inference, and within the nonparametric learning framework. The MMD is commonly estimated at a root-$m$ rate, where $m$ is the number of simulated samples. This can lead to significant computational challenges since a large $m$ is required to obtain an accurate estimate, which is crucial for parameter estimation. In this paper, we propose a novel estimator for the MMD with significantly improved sample complexity. The estimator is particularly well suited for computationally expensive smooth simulators with low- to mid-dimensional inputs. This claim is supported through both theoretical results and an extensive simulation study on benchmark simulators.
\end{abstract}

\section{Introduction}

Many domains of science, medicine and engineering use our mechanistic understanding of real-world phenomena to create simulators that can represent system behaviour in different circumstances. Such simulator-based models define a stochastic procedure that can generate (possibly complex) synthetic data-sets, and are widely used in fields such as population genetics \cite{Beaumont2010}, ecology \cite{Wood2010}, astronomy \cite{Cameron2012, Akeret_2015}, epidemiology \cite{Kypraios2017}, atmospheric contamination \cite{Kopka2016}, radio propagation \cite{Bharti2022}, and agent-based modelling \cite{Jennings1999}. However, the ease of simulating data from the model comes at the cost of an intractable likelihood function, rendering most standard statistical inference methods inapplicable to such models. To solve this issue, a host of \emph{likelihood-free inference} methods have been developed that circumvent the need to evaluate the likelihood or its derivatives, see \citet{Lintusaari2016, Cranmer2020} for an overview.

\begin{figure}
	 \centering
	 \includegraphics[trim={0 10 70 5}, clip,  width = 0.9\linewidth]{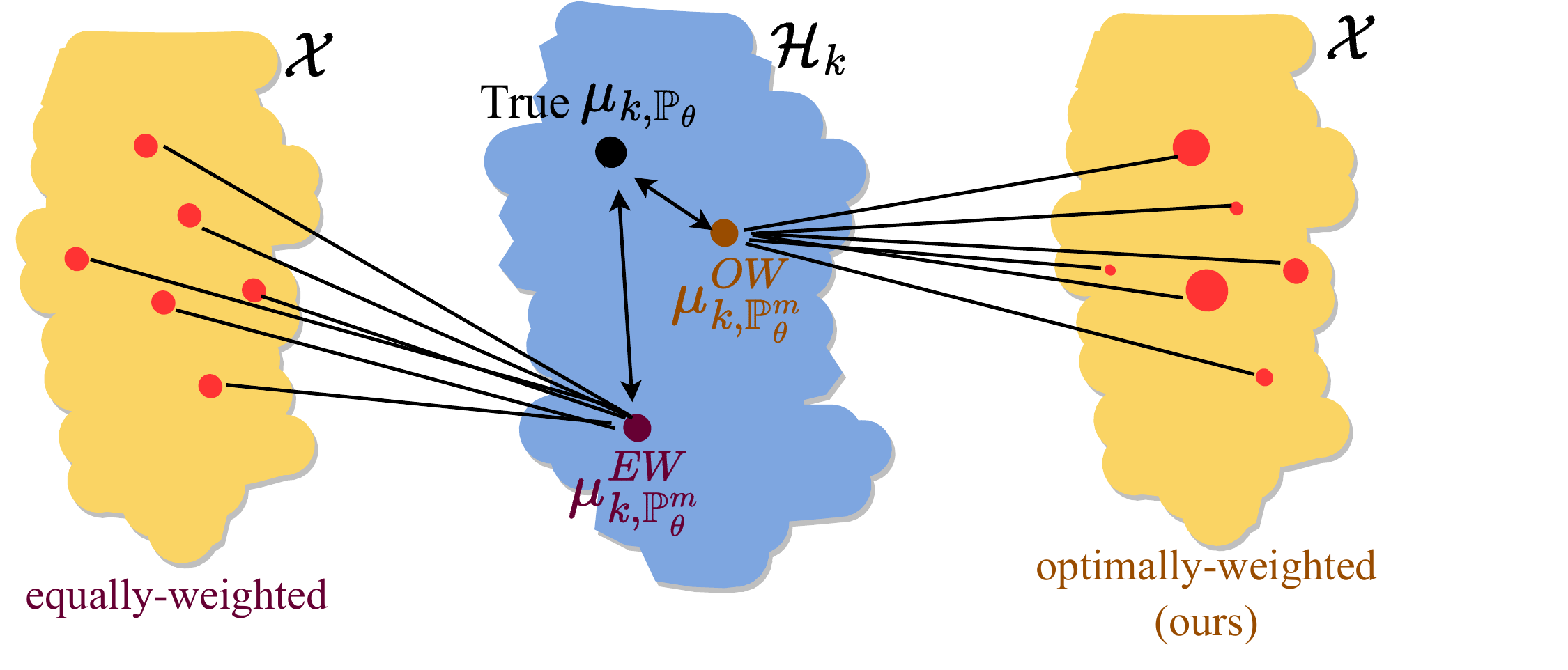}
\caption{
Estimating the MMD requires approximating the embedding $\mu_{k,\P_\theta}$ of the model $\P_\theta$ in a reproducing kernel Hilbert space $\mathcal{H}_k$. The classical approach consists of doing this from $m$ \textcolor{purple}{equally-weighted} independent samples from $\P_\theta$ (denoted \textcolor{purple}{$\mu^{EW}_{k,\P_\theta^m}$}), but we show in this paper that it is possible to improve this estimator by using \textcolor{orange}{optimally-weighted} samples (denoted \textcolor{orange}{$ \mu^{OW}_{k,\P_\theta^m}$}).  \vspace{-3mm}
}
\label{fig:sketch}
\end{figure}

A common approach for likelihood-free inference involves comparing simulated observations from the model and the observed data, with respect to some notion of distance. Accurately estimating the distance is essential for inference but doing so usually requires simulating large amounts of synthetic data. This can be a computational bottleneck, especially for expensive simulators, which in the most extreme cases can take up to hundreds or thousands of CPU hours per simulation; see  \citet{Niederer2019} for an example in cardiac modelling. Other examples include tsunami models based on shallow water equations that require several GPU hours per run \cite{Behrens2015}, runaway electron analysis models for nuclear fusion devices that require 24 CPU hours per run \cite{DREAM}, and models of large-scale wind farms that require 100 CPU hours per run \cite{Kirby2022}. 
Naturally, the discrepancies popular for likelihood-free inference are those which can be efficiently estimated given samples from two distributions, such as the KL divergence \cite{Jiang2018},  Wasserstein distance \cite{Peyre2019, Bernton2019}, Sinkhorn divergence \cite{Genevay2017}, energy distance \cite{Nguyen2020}, classification accuracy \cite{Gutmann2017}, or the maximum mean discrepancy, the latter of which is the topic of this paper. Here, ``efficiently estimated'' is defined in terms of \emph{sample complexity}, which is the rate of convergence at which a statistical distance can be estimated from samples. The faster an estimator converges in the number of samples, the less we need to simulate from the model, and hence, the smaller the computational cost.

We focus on the maximum mean discrepancy (MMD) \cite{Gretton2006, Gretton2012JMLR}, a probability metric which measures the distance between distributions through the distance between their embeddings in a reproducing kernel Hilbert space; see \Cref{fig:sketch} for an illustration. A number of advantages of this distance are commonly put forward in the literature: (i) it has relatively low sample complexity when compared to its alternatives listed above, (ii) it has desirable statistical properties, such as leading to consistent and robust estimators, (iii) it is applicable on any data-type for which a kernel can be defined, and does not require hand-crafted summary statistics. Due to these attractive properties, the MMD has been used in a range of frameworks for likelihood-free inference, including for approximate Bayesian computation (ABC) \cite{Park2015, Mitrovic2016, Kajihara2018, Bharti2022,Legramanti2022}, for minimum distance estimation (MDE) \cite{Briol2019MMD, cherief2021, Alquier2021,Niu2021,Key2021}, for generalised Bayesian inference  \cite{cherief2020mmd,Pacchiardi2021}, for Bayesian nonparametric learning \cite{Dellaporta2022}, and for training generative adversarial networks \cite{Dziugaite2015,Li2015GMMN,Li2017MMDGAN,Binkowski2018}.

In this paper, we do not revisit the question of whether the MMD is the best choice of distance for a particular problem. Instead, we assume that the MMD has been chosen, and focus on constructing estimators with strong sample complexity for this distance. The most common estimators for the MMD are U-statistic or V-statistic estimators, and these have sample complexity of $\mathcal O(m^{-\frac{1}{2}})$, under mild conditions \cite{Briol2019MMD}, where $m$ is the number of samples. In recent work, \citet{Niu2021} showed that this can be improved to $\mathcal O(m^{-1+\epsilon})$ for any $\epsilon>0$ through the use of a V-statistic estimator and randomised quasi-Monte Carlo (RQMC) sampling. This significant improvement does come at the cost of restrictive assumptions --- the simulator must be written in a form where the inputs are uniform random variables, and must satisfy stringent smoothness conditions which are difficult to verify in practice.

In this paper, we propose a novel set of \emph{optimally-weighted} estimators with sample complexity of  $\mathcal O(m^{-\frac{\nu_c}{s} - \frac{1}{2}})$ where $s$ is the dimension of the base space and $\nu_c $ is a parameter depending on the smoothness of the kernel and the simulator. This leads to significantly improved sample complexity against both U- or V-statistic and independent samples for any $\nu_c$, and against RQMC when $\nu_c>s/2$. Additionally, the optimality of the weights guarantees that even if this condition is not satisfied, the order of the sample complexity is still at least as good as that for existing estimators.

The remainder of the paper is structured as follows. \Cref{sec: background} recalls existing estimators for the MMD, and how these are used in likelihood-free inference. \Cref{sec: methodology} presents our estimators, and \Cref{sec:theory} provides a theoretical analysis of their sample complexity. Finally, \Cref{sec:experiments} demonstrates strong empirical performance on a range of simulators, and \Cref{sec:conclusion} discusses future research.

\section{Background}\label{sec: background}

Throughout the paper, $\X$ will denote some set, and $\PX$ will be the set of all Borel probability measures on $\X$.

\vspace{-2mm}

\paragraph{Likelihood-free inference.}
We consider the classic parameter estimation problem, where we assume that we observe some independent and identically distributed (iid) realisations $\{x_i\}_{i=1}^n \subseteq \X$ from some data-generating mechanism $\Q \in \PX$. Given $\{x_i\}_{i=1}^n$ and a  parametric family of distributions $ \{ \P_\theta : \theta \in \Theta\} \subset \PX$ (i.e. the model) with parameter space $\Theta$, we are interested in recovering the parameter value $\theta^* \in \Theta$ such that $\P_{\theta^*}$ is either equal, or in some sense closest, to $\Q$. 

The challenge in likelihood-free inference is that the likelihood associated with $\P_\theta$ is intractable, meaning it cannot be evaluated pointwise. This prevents the use of classical methods such as maximum likelihood estimation or (exact) Bayesian inference. Instead, we assume that we are able to simulate iid realisations from $\P_\theta$, and such models are hence called generative models or simulator-based models. Such models are characterised through their generative process, a pair $(G_\theta,\mathbb{U})$ consisting of a simple distribution $\U$ (such as a multivariate Gaussian or uniform distribution) on a space $\mathcal{U}$ and a map $G_\theta: \mathcal{U} \rightarrow \mathcal{X}$ called the generator or simulator. We will call $\U$ a base measure and $\mathcal{U}$ the base space, and consider $\calU \subset \R^s$ and $\mathcal{X} \subseteq \R^d$.
To sample $y \sim \P_\theta$, one can first sample $u \sim \U$, then apply the generator $y = G_\theta(u)$. To perform parameter estimation for these models, it is common to repeatedly sample simulated data from the model for different parameter values and compare them to $\{x_i\}_{i=1}^n$ using a distance. We now recall the distance which will be the focus of this paper.

\vspace{-2mm}

\paragraph{Maximum mean discrepancy (MMD).} 
Let $\Hk$ be a reproducible kernel Hilbert space (RKHS) associated with the symmetric and positive definite function $k: \X \times \X \rightarrow \R$ \cite{Berlinet2004}, called a reproducing kernel, and denote by $\|\cdot\|_{\Hk}$ and $\langle \cdot,\cdot \rangle_{\Hk}$ the corresponding norm and inner product. Additionally, let $\Pk := \{ \P \in \PX : \int_\X \sqrt{k(x,x)} \P(\text{d}x) < \infty \}$; whenever $k$ is bounded, $\P_k(\X)=\P(\X)$. As illustrated in the sketch in \Cref{fig:sketch}, any distribution $\P \in \Pk$ can be mapped into $\mathcal{H}_k$ via its kernel mean embedding, defined as $\mu_{k,\P} = \int_\X k(\cdot, x) \P(\text{d}x)$. Then, the MMD between $\P$ and $\Q$ is the distance between their embeddings in $\Hk$:
\begin{talign}\label{eq:MMD_withembeddings}
 \text{MMD}_k(\P, \Q) = \Vert \mu_{k,\P} - \mu_{k,\Q} \Vert_{\Hk},
\end{talign}
see \citet{Muandet2017} for a review. Alternatively, the MMD can also be expressed as $\text{MMD}_k (\P, \Q) 
    = \sup_{\|f\|_{\Hk}\leq 1} \left| \int_{\X} f(x) \P(\text{d}x) - \int_{\X} f(x) \Q(\text{d}x) \right|,$
where the supremum is taken over all the functions in the unit-ball of the RKHS $\mathcal{H}_k$. 
Whenever $k$ is a characteristic kernel, the MMD is a probability metric, meaning that $\text{MMD}_k(\P, \Q) = 0$ if and only if $\P = \Q$. This condition is satisfied for kernels including the squared-exponential (SE) $k_\mathrm{SE}(x,y) = \eta \exp(-\|x-y\|^2_2/l^2)$, the Mat\'ern $k_\nu(x,y) = \frac{\eta}{\Gamma(\nu)2^{\nu - 1}} (\frac{\sqrt{2 \nu}}{l} \| x - y \|_2 )^\nu K_\nu(\frac{\sqrt{2 \nu}}{l} \| x-y \|_2)$, where $K_\nu$ is the modified Bessel function of the second kind, and the inverse-multiquadric kernels on $\X = \R^d$ \cite{Sriperumbudur2009}. Mat\'ern kernels are of particular interest: the \emph{order} parameter $\nu$ uniquely determines the smoothness of $\Hk$, and for half-integer orders $\nu \in \{\frac{1}{2}, \frac{3}{2}, \dots\}$, the kernel $k_\nu$ can be written as a product of an exponential and a polynomial of order $\lfloor \nu \rfloor$~\citep{Rasmussen2006}.

Unfortunately, the expression in \eqref{eq:MMD_withembeddings} usually cannot be computed directly since $\mu_{k,\P}$ will not be available in closed form outside of a limited number of $(k,\P)$ pairs. Instead, using the reproducing property (i.e. $f(x) = \langle f, k(\cdot,x)\rangle_{\Hk}$ $\forall f \in \Hk$), we can write
\begin{talign}\label{eq:MMD2_exact}
     \text{MMD}^2_k(\P,\Q) &= \int_{\X} \int_{\X}  k(x,y) \P(\text{d}x) \P(\text{d}y)  \nonumber\\
    & - 2 \int_{\X} \int_{\X}  k(x,y) \P(\text{d}x) \Q(\text{d}y)  \nonumber\\
    & + \int_{\X} \int_{\X}  k(x,y) \Q(\text{d}x) \Q(\text{d}y).
\end{talign}
This expression is convenient to work with as it can be estimated through approximations of the integrals. Let $\{y_i\}_{i=1}^m \sim \P$, $\{x_i\}_{i=1}^n \sim \Q$ and let $\P^m = \frac{1}{m} \sum_{j=1}^m \delta_{y_j}$ and $\Q^n = \frac{1}{n} \sum_{i=1}^n \delta_{x_i}$, where $\delta_{x_i}$ is a Dirac measure at $x_i$. The squared-MMD can be approximated through a V-statistic as
\begin{talign}
    &\text{MMD}_k^2(\P^m,\Q^n) = \frac{1}{m^2} \sum_{i, j = 1}^m k(y_i,y_j) \nonumber\\
    &- \frac{2}{nm} \sum_{i=1}^n \sum_{j=1}^m k(x_i,y_j) +  \frac{1}{n^2} \sum_{i, j = 1}^n k(x_i,x_j).\nonumber
\end{talign}
This is equivalent to approximating $\mu_{k,\P}$ using $\mu^{\text{EW}}_{k,\P^{m}}(x) = \frac{1}{m} \sum_{i=1}^m k(x,x_i)$.
Alternatively, one can use an unbiased U-statistic approximation \cite{Gretton2012JMLR}.
Both of these estimates can be calculated straightforwardly via evaluations of the kernel $k$ at a computational cost~$\mathcal O(m^2+mn+n^2)$.

\vspace{-2mm}

\paragraph{Likelihood-free inference with the MMD.} The MMD has been used within a range of frameworks. 
In a frequentist setting, the MMD was proposed for minimum distance estimation by \citet{Briol2019MMD}:
\begin{talign}
    \hat{\theta}_n = \underset{\theta \in \Theta}{\arg \min}~ \text{MMD}_k^2(\P_\theta,\Q^n).
    \label{eq:MMD_estimator}
\end{talign}
In practice, the minimiser is computed through an optimisation algorithm, which requires evaluations of the squared-MMD or of its gradient. Such evaluations are intractable, but any estimator can be used within a stochastic optimisation algorithm. Similar optimisation problems and stochastic approximations also arise when using the MMD for generative adversarial networks  \cite{Dziugaite2015,Li2015GMMN} and for nonparametric learning  \cite{Dellaporta2022}.

In a Bayesian setting, the MMD has been used to create several pseudo-posteriors by updating a prior distribution $p$ on $\Theta$ using data. For example, the K2-ABC posterior of \citet{Park2015} is a pseudo-posterior of the form:
\begin{talign}
    p_{\mathrm{ABC}}(\theta | x_1\dots, x_n) \propto  \int &\dots \int  \Pi_{j=1}^m \mathbbm{1}_{\{\text{MMD}^2_k(\P_\theta,\Q^n) < \varepsilon\}}(\theta) \nonumber \\ 
    &p(y_j|\theta) p(\theta) \text{d}y_1, \dots, \text{d}y_m. \label{eq:ABC}
\end{talign}
where the indicator function $\mathbbm{1}_{\{A\}}$ is equal to 1 if event $A$ holds.
Here, the MMD is used to determine whether a particular instance of the parametric model is within an $\varepsilon$ distance from the data. The K2-ABC algorithm approximates this pseudo-posterior through sampling of the model $\P_\theta$ which leads to the use of an estimator of the squared-MMD.

Finally, the MMD has also been used for generalised Bayesian inference, where it is used to construct the MMD-Bayes posterior \cite{cherief2020mmd}
\begin{talign*}
    p_{\mathrm{GBI}}(\theta | x_1\dots, x_n) \propto  \exp( - \text{MMD}_k^2(\P_\theta,\Q^n)) p(\theta).
\end{talign*}
Once again, this pseudo-posterior is intractable, but it can be approximated through pseudo-marginal MCMC, in which case an unbiased estimator is used in place of the squared-MMD \cite{Pacchiardi2021}.

\vspace{-2mm}

\paragraph{Sample complexity of MMD estimators.}  \label{sec:complexity_bg}

As highlighted above, the performance of these likelihood-free inference methods relies heavily on how accurately we can estimate the MMD using samples; that is, how fast our estimator approaches $\text{MMD}_k(\P_\theta,\Q)$ as a function of $n$ and $m$, the number of observed and simulated data points, respectively. Let $\widehat{\text{MMD}}_k(\P_\theta^m,\Q^n)$ be any estimator of the MMD based on $m$ simulated data points. Using the triangle inequality, this error can be decomposed as follows: 
\begin{talign}
    &|\text{MMD}_k(\P_\theta,\Q)-\widehat{\text{MMD}}_k(\P_\theta^m,\Q^n)| \nonumber \\
    & \leq  
    |\text{MMD}_k(\P_\theta,\Q)-\text{MMD}_k(\P_\theta,\Q^n)| \nonumber \\
    & \; + |\text{MMD}_k(\P_\theta,\Q^n)- \widehat{\text{MMD}}_k(\P_\theta^m,\Q^n)| \label{eq:error_decomposition}
\end{talign}
where the first term describes the approximation error due to having a finite number of data points $n$, and the second term describes the error due to a finite number $m$ of simulator evaluations. To understand the behaviour of the first term, we can use the following sample complexity result for the V-statistic. The proof is a direct application of the triangle inequality together with Lemma 1 in \cite{Briol2019MMD}. 
\begin{theorem} Suppose that $\sup_{x,x'} k(x,x') <\infty$ and let $\Q^n$ consist of $n$ iid realisations from $\Q \in \Pk$. Then, for any $\P \in \Pk$, we have with high probability 
\begin{talign*}
    \left| \text{MMD}_k(\P,\Q) - \text{MMD}_k(\P,\Q^n)\right| =  \mathcal O(n^{-\frac{1}{2}}).
\end{talign*}
\end{theorem}
When $\widehat{\text{MMD}}_k(\P_\theta^m,\Q^n)$ is also a V-statistic approximation, both terms in \eqref{eq:error_decomposition} can be tackled with this result and the overall error is of size $\mathcal O(n^{-\frac{1}{2}}+m^{-\frac{1}{2}})$. This shows that we should take $m = \mathcal O(n)$ to ensure a good enough approximation of the MMD. Though this rate has the advantage of being independent of the dimension of $\X$, it is relatively slow in $m$. We therefore require a large number of simulated data points, which can be computationally expensive.

\citet{Niu2021} recently proposed an alternate approach based on randomised quasi-Monte Carlo (RQMC) \cite{Dick2013} samples within a V-statistic. Using stronger assumptions on $\U, k$ and $G_\theta$, they are able to obtain an estimator with improved sample complexity. We now state their assumptions and result below.

For $f:\X \to \R$ and a multi-index $\alpha = (\alpha_1, \dots \alpha_d) \in \mathbb{N}^d$, we denote the $|\alpha|=\sum_{i=1}^d \alpha_i$ order partial derivative $\partial^\alpha f=\partial^{|\alpha|}f/\partial^{\alpha_1} x_1 \ldots \partial^{\alpha_d}x_d$ by $\partial^\alpha f$. We say $f \in C^{m}(\X)$, for $m \in \mathbb{N}$, if $\partial^\alpha f$ exists and is continuous for any $|\alpha| \in [0, m]$. For two-variable $f:\X \times \X \to \R$, $\partial^{\alpha,\alpha} f$ is the $\alpha$-partial derivative in each variable. The norm $\|\cdot\|_{\L^p(\X)}$ for $f:\X \to \R$ is defined as $\|f\|_{\L^p(\X)} = (\int_\X |f(x)|^p \d x)^{1/p}$. The notation $a_{v}:b_{-v}$ represents a point $u\in[a,b]^{s}$ with $u_{j}=a_{j}$ for $j\in v$, and $u_{j}=b_{j}$ for $j\notin v$.

\begin{oldassumption}
\label{as:oldpoints}
    The base space $\calU = [0,1]^s$, the base measure $\U$ is uniform on $\calU$, and $\{u_i\}_{i=1}^m \subset \calU$ forms an RQMC point set.
\end{oldassumption}

\begin{oldassumption}
\label{as:oldgenerator}
    The generator $G_\theta:[0,1]^s \rightarrow \X$ is s.t.: 
\vspace{-2ex}
\begin{enumerate}
  \setlength{\itemsep}{1pt}
  \setlength{\parskip}{0pt}
  \setlength{\parsep}{0pt}
     \item $\partial^{(1,\ldots,1)} G_{\theta,j} \in C([0,1]^s)$ for all $j = 1, \ldots, d$. 
    \item for all $j = 1, \ldots, d$ and  $v \in \{0,1\}^s\setminus(0,\ldots,0)$, there is a $p_j \in [1,\infty]$, $\sum_{j=1}^d p_j^{-1} \leq 1$, such that for $g(\cdot)=\partial^{v} G_{\theta,j}(\cdot:1_{-v})$ it holds that $\|g\|_{\L^{p_j}([0,1]^{|v|})} < \infty$.
\end{enumerate}
\end{oldassumption}
\begin{oldassumption}
\label{as:oldkernels}
    For any $x \in \X$, $k(x, \cdot) \in C^s(\X)$ and $\forall t\in\mathbb{N}^{d},|t|\leq s,\sup_{x\in\X}\partial^{t,t}k(x,x)<C_k$ where $C_k$ is some universal constant depending only on $k$.
\end{oldassumption}
\begin{theorem}
Under~\Cref{as:oldpoints,as:oldgenerator,as:oldkernels} and $\Q \in \Pk$,
\begin{talign*}
    \left| \text{MMD}_k(\P_\theta,\Q) - \text{MMD}_k(\P_\theta^m,\Q)\right| =  \mathcal O(m^{-1+\epsilon}).
\end{talign*}
\end{theorem}

In this case, the second term in \eqref{eq:error_decomposition} decreases at a faster rate than the first term and the overall error decreases as $\mathcal O(n^{-\frac{1}{2}}+m^{-1+\epsilon})$ for any $\epsilon>0$. As a result, (ignoring log-terms) we can take $m = \mathcal O(n^{-\frac{1}{2}})$, meaning a much smaller number of simulations are required.
However, the technical conditions required are either very restrictive ($\U$ must be uniform), or will be difficult to verify in practice (the conditions on $G_\theta$ are not very interpretable and difficult to verify). Hence, the range of cases where RQMC can be applied is limited. Additionally, when both $k$ and $G_\theta$ are smooth, faster rates can  be obtained using our optimally-weighted estimator presented in the next section.

\section{Optimally-Weighted Estimators}\label{sec: methodology}

We now present our estimator, which weighs simulated data.
To that end, we denote the empirical measure of the simulated data as $\P_\theta^{m,w} = \sum_{i=1}^m w_i \delta_{y_i}$ where $y_i = G_\theta(u_i)$, and $w_i \in \R$ is the weight associated with $y_i \in \X$ for all $i \in \{1,\ldots,m\}$.  Assuming for a moment that these weights are known, then we have 
\begin{talign} \label{eq:optimally_weighted_MMD2}
    &\text{MMD}^2_k(\P_\theta^{m,w},\Q^n) = \sum_{i, j = 1}^m w_i w_j k(y_i,y_j) \\
    &- \frac{2}{n} \sum_{i=1}^n \sum_{j=1}^m w_j k(x_i,y_j) + \frac{1}{n^2} \sum_{i, j = 1}^n k(x_i,x_j).\nonumber
\end{talign}
Clearly, $w_i=1/m$ for all $i$ recovers the V-statistic approximation of the squared-MMD, but here we have additional flexibility in how to select these weights and not impose any constraints on them beyond being real-valued. To identify our choice of weights, we will make use of a tight upper bound on the approximation error. 
\begin{theorem}\label{thm:optimal_weights}
Let $c:\calU \times \calU \rightarrow \R$ be a reproducing kernel such that $k(x, \cdot) \circ G_\theta \in \mathcal{H}_c$ and $\Q \in \Pk$.
Then, $\exists K>0$ independent of $\{u_i,y_i,w_i\}_{i=1}^m$ but dependent on $c, k$ and $G_\theta$ such that:
\begin{talign*}
    | \text{MMD}_k(\P_\theta,\Q) - \text{MMD}_k(\P_\theta^{m,w},\Q)| \\ \leq K \times \text{MMD}_c \left(\mathbb{U}, \sum_{i=1}^m w_i \delta_{u_i}\right),
\end{talign*}
Additionally, the weights minimising this upper bound can be obtained in closed-form; i.e. 
\begin{talign} 
    w^* & = \underset{ w \in \R^m}{\arg \min} ~ \text{MMD}_c\left(\U,\sum_{i=1}^m w_i \delta_{u_i}\right)\nonumber  \\
    & =  c(U,U)^{-1}  z(U)\label{eq:optimal_weights}
\end{talign}
 where $z(U)_i = \mu_{c,\U}(u_i) =  \int_{\calU} c(u_i,u) \U(\text{d}u)$ is the kernel mean embedding of $\U$ in the RKHS $\mathcal{H}_c$ and $(c(U,U))_{ij} = c(u_i,u_j)$ for all $i,j \in \{1,\ldots,m\}$.
\end{theorem}
Our \textit{optimally-weighted (OW) estimator} is the weighted estimator in \eqref{eq:optimally_weighted_MMD2}
with the optimal weights in \eqref{eq:optimal_weights}. This corresponds to estimating $\mu_{k,\P_\theta}$ with a weighted approximation $\mu^{\text{OW}}_{k,\P^{m}_\theta} = \sum_{i=1}^n w^*_i k(x,x_i) = \sum_{i=1}^n w^*_i k(x,G_\theta(u_i))$ where $w_i^*$ represents the importance of $x_i = G_\theta(u_i)$ for our approximation. 
To calculate these weights, we need to evaluate $\mu_{c,\U}$ pointwise in closed-form. The key insight is that although $\mu_{k,\P_\theta}$ will usually not be available in closed-form,  the same is not true for $\mu_{c,\U}$. This is because, unlike $\P_\theta$, $\U$ is usually a simple distribution such as a uniform, Gaussian, Gamma or Poisson. Additionally, we have full flexibility in our choice of $c$ so long as $k(x, \cdot) \circ G_\theta \in \mathcal{H}_c$. We refer to Table~1 in \cite{Briol2019PI} or the \texttt{ProbNum} Python package \cite{Wenger2021} for a list of known closed-form kernel embeddings. 

The proof of this result (see \Cref{app:proofs_optimal_weights}) relies on two inequalities which make the overall result tight. The first  is a reverse triangle inequality, which allows us to remove dependence on the true data-generating distribution $\Q$, a quantity which is always unknown to us. In this sense, the bound is ``worst-case optimal'' over $\Q$, a desirable property for likelihood-free inference. The second inequality allows us to use the kernel $c$ instead of $c_\theta(u,v) = k(G_\theta(u), G_\theta(v))$ to construct our weights. Of course, this bound is attained if $c=c_\theta$ and is therefore tight. 
However, in practice this choice will often be infeasible due to lack of closed-form kernel embeddings $\mu_{c_\theta,\mathbb{U}}$. We therefore choose $c$ such that the RKHS it induces contains the RKHS induced by $c_\theta$. At a high-level, the smaller the gap between these spaces, the better the bound will be. This choice of $c$ will be explored further through theory (in \Cref{sec:theory}) and experiments (in \Cref{sec:experiments}).

\vspace{-2mm}

\paragraph{Related methods.} The optimal weights in \Cref{thm:optimal_weights} are equivalent to Bayesian quadrature (BQ) weights \cite{Diaconis1988,OHagan1991,Rasmussen2003,Briol2019PI}. BQ is a method for numerical integration based on Gaussian process regression (in our case with prior mean zero and prior covariance function $c$). We can therefore think of our estimator as performing BQ to approximate all integrals against $\P$ in \eqref{eq:MMD2_exact}. This interpretation is helpful for selecting $c$ --- the kernel should be chosen so that the corresponding Gaussian process is a good prior for the integrands in \eqref{eq:MMD2_exact}. This correspondence will also help us derive sample complexity results in the next section.

Our estimator minimises $\text{MMD}_c \left(\mathbb{U}, \sum_{i=1}^m w_i \delta_{u_i}\right)$ over the choice of weights, but we also have flexibility over the choice of $\{u_i\}_{i=1}^m$. Unfortunately, this optimisation cannot be solved in closed-form, and is in fact usually not convex. There is a wide range of methods which have been proposed to do point selection so as to minimise an MMD with equally-weighted points. Kernel thinning \cite{Dwivedi2021}, support points \cite{Mak2018} and Stein thinning \cite{Riabiz2020} are methods based on the MMD to subsample points given a large dataset. Kernel herding \cite{Chen2010,Bach2012} and Stein points \cite{Chen2018,Chen2019} are sequential point selection methods which use an MMD as objective. In addition, similar point selection methods have also been proposed for BQ \cite{Gunter2014,Briol2015,Bardenet2019} and these are therefore closest to our OW setting.

\begin{table*}[t!]
\centering
\caption{Average and standard deviation (in parenthesis) of estimated MMD$^2 ~(\times 10^{-3})$ between $\P_{\theta}^m$ and $\P_{\theta}^n $  computed over 100 runs for the V-statistic and our optimally-weighted (OW) estimator. Settings: $n = 10,000$, $m = 256$. }
\resizebox{0.98\textwidth}{!}{%
\begin{tabular}{@{}c  lll | ll | ll@{}}
\toprule
\multicolumn{1}{c}{Model} & $s$ & $d$ & References & IID V-stat & IID OW (ours) & RQMC V-stat & RQMC OW (ours)\\ \midrule
g-and-k                   & 1                            & 1                                 & \cite{Bharti22a,Niu2021}                                              & 2.25 \textcolor{gray}{(1.52) }                                                                                & \textbf{0.086} \textcolor{gray}{(0.049)}                                                                           & 0.060 \textcolor{gray}{(0.037)}                                                                                & \textbf{0.059} \textcolor{gray}{(0.037) }                                                                         \\
Two moons               & 2         & 2           & \cite{Lueckmann2021,Wiqvist2021}  & 2.36 \textcolor{gray}{(1.94)}                      & \textbf{0.057} \textcolor{gray}{(0.054)}                  & 0.056 \textcolor{gray}{(0.044)}                  & \textbf{0.055} \textcolor{gray}{(0.044)}         \\
Bivariate Beta            & 5                                                                                        & 2                                &    \cite{Nguyen2020,Niu2021}                                            & 2.13 \textcolor{gray}{(1.17)}                                                                                 & \textbf{0.555} \textcolor{gray}{(0.227)}                                                                           & 0.222 \textcolor{gray}{(0.111)}                                                                                & \textbf{0.193} \textcolor{gray}{(0.088)}                                                                            \\
MA(2)        & 12                                                                                       & 10                                            &    \cite{Marin2011,Nguyen2020}                                & 2.42 \textcolor{gray}{(0.80)}                                                                                 & \textbf{0.705} \textcolor{gray}{ (0.107) }                                                                           & 0.381 \textcolor{gray}{(0.054) }                                                                               & \textbf{0.322} \textcolor{gray}{(0.052) }                                                                           \\
M/G/1 queue            & 10                                                                                       & 5                                &                                       \cite{Pacchiardi2021, Jiang2018}          & 2.52 \textcolor{gray}{(1.19) }                                                                               & \textbf{1.71} \textcolor{gray}{(0.568)}                                                                            & \textbf{0.595} \textcolor{gray}{(0.134)}                                                                                    & 0.646 \textcolor{gray}{(0.202)}                                                                         \\
Lotka-Volterra            & 600                                                                                      & 2                                &   \cite{Briol2019MMD, Wiqvist2021}                                              & 2.13 \textcolor{gray}{(1.10)}                                                                                & \textbf{2.04} \textcolor{gray}{(0.956)}                                                                            & 1.44 \textcolor{gray}{(0.955)}                                                                                           & \textbf{1.42} \textcolor{gray}{(0.942)}                                                                                       \\ \bottomrule
\end{tabular}
}
\label{tab:MMDerror}
\end{table*}

\section{Theoretical Guarantees} \label{sec:theory}

\paragraph{Sample complexity.} The following theorem establishes a sample complexity of $\mathcal O(m^{-\frac{\nu_c}{s} - \frac{1}{2}})$ for our optimally-weighted estimator, where $\nu_c$ is a parameter depending on the smoothness of $k$ and $G_\theta$. We achieve a better rate than RQMC under milder conditions, as discussed below.

\begin{assumption}
\label{as:points}
    The base space $\calU \subset \R^s$ is bounded, open, and convex, the data space $\X$ is the entire $\R^d$ or is bounded, open, and convex. The base measure $\U$ has a density $f_\U: \calU \to [C_\U, C_\U']$ for some $C_\U$, $C_\U'>0$, and $\P_\theta$ has a density bounded above. 
    The point set $\{u_i\}_{i=1}^m \subset \calU$ has a fill distance of asymptotics $h_m = \mathcal O(m^{-\frac{1}{s}})$, where $h_m = \sup_{u \in \calU} \min_{i \in [1, m]} \|u - u_i\|_2$.
\end{assumption}

Our assumptions on  $\calU$ and $\U$ are milder than those of~\Cref{as:oldpoints}, which requires $\U$ to be uniform. The assumptions on $\X$ and $\P_\theta$ are likely to hold for simulators in practice. 
We replace the requirement that the point set $\{u_i\}_{i=1}^m$ is RQMC with a milder assumption on the fill distance, which quantifies how far any point in $\calU$ can get from the set $\{u_i\}_{i=1}^m$. The fill distance asymptotics is a standard assumption that ensures the coverage of $\calU$; for example, it holds for regular grids, and in expectation for independent samples. For further examples of point sets that guarantee small fill distance, see \citet{Wynne2020}. 

\begin{assumption}
\label{as:generator}
    The generator is a map $G_\theta:\calU \rightarrow \X$ such that for some integer $l>s/2$, any $j \in [1, d]$ and any multi-index $\alpha \in \mathbb{N}^d$ of size $|\alpha|\leq l$, the partial derivative $\partial^\alpha G_{\theta,j}$ exists and is bounded from above.
\end{assumption}
Assumption \Cref{as:generator} is more interpretable and easier to check than~\Cref{as:oldgenerator} (specifically part 2) as it just requires knowing how many derivatives $G_\theta$ has.
As stated in~\citet{Niu2021}, a simpler condition that implies~\Cref{as:oldgenerator} needs $G_\theta$ to be smooth up to the order $l \geq s$, which rules out the standard choices of $\nu \in \{\frac{1}{2}, \frac{3}{2}, \frac{5}{2}\}$ for large enough $s$. In contrast, we only ask that $l > s/2$. 

\begin{assumption}
\label{as:kernels}
    $k$ is a Mat\'ern kernel on $\X$ of order $\nu_k$ such that $\lfloor \nu_k + d/2 \rfloor > s/2$, or an SE kernel, and $c$ is a Mat\'ern kernel on $\calU$ of order $\nu_c \leq \min(\lfloor \nu_k + d/2 \rfloor, l)$.
\end{assumption}

\Cref{as:kernels} places less restrictions on the choice of $k$ than \Cref{as:oldkernels}. Although both allow for $k$ to be the SE kernel, as a corollary of the Sobolev embedding theorem~\citep[Theorem 4.12]{adams2003sobolev},~\Cref{as:oldkernels} only holds for a Mat\'ern $k$ if $\lceil \nu_k \rceil \geq s + 1$ (i.e. smooth $k$), while our lower bound on $\nu_k$ is much less restrictive. The conditions on $c$ are needed to ensure $k(x, \cdot) \circ G_\theta \in \Hc$. Note that these could be weakened using the work of \cite{Kanagawa2017,Teckentrup2020,Wynne2020}, but at the expense of more restrictive conditions on $\{u_i\}_{i=1}^m$ in \Cref{as:points}.
\begin{theorem}\label{thm:rate_of_convergence}
Under \Cref{as:points,as:generator,as:kernels}, $k(x, \cdot) \circ G_\theta \in \Hc$ holds, and for any  and $\Q \in \Pk$:
\begin{talign*}
    \left| \text{MMD}_k(\P_\theta,\Q) - \text{MMD}_k(\P_\theta^{m,w},\Q)\right| = \mathcal O(m^{-\frac{\nu_c}{s} - \frac{1}{2}}). 
\end{talign*}
\end{theorem}
The result shows that our method has improved sample complexity over the V-statistic for any $\nu_c$ and $s$. Additionally, it is better than RQMC when $\nu_c>s/2$. In practice, we should pick a kernel $c$ that is as smooth as possible whilst not being smoother than $G_\theta$ or $k$, as per \Cref{as:kernels}. 
Hence, we should take $\nu_c$ to be smaller than $l$ and $\nu_k$, the smoothness of $G_\theta$ and $k$, respectively. 
In case the smoothness of $G_\theta$ is unknown, the conservative choice is to take a smaller value of $\nu_c$ to ensure \Cref{as:kernels} is satisfied. 

\begin{figure}
	 \centering
	 \includegraphics[trim={40 615 255 45}, clip,  width = 0.75\linewidth]{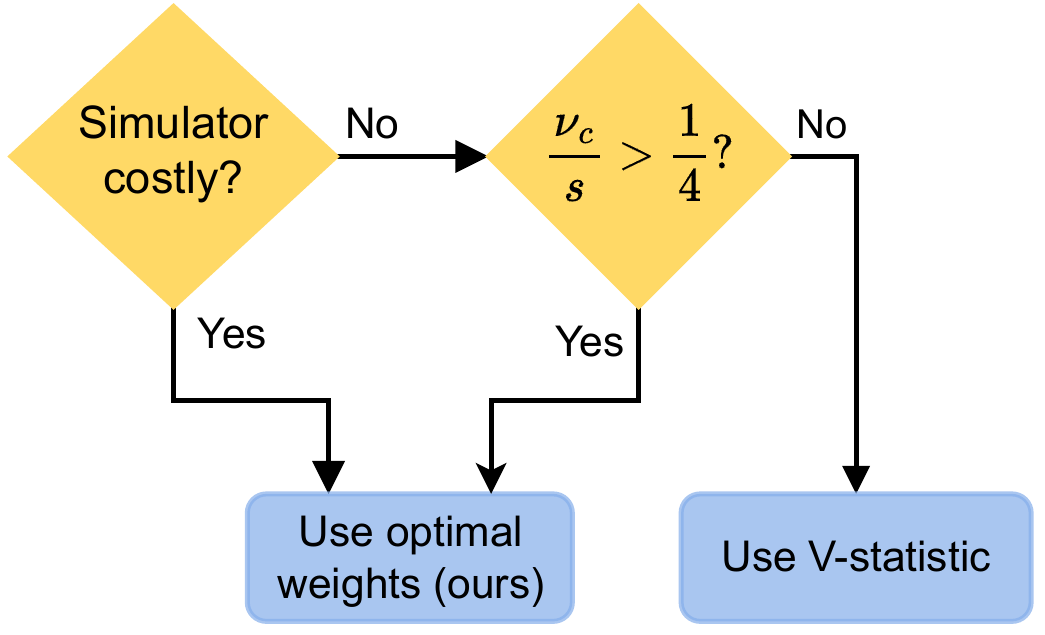}
\caption{Guidelines on when to use our optimally-weighted estimator over the V-statistic: a) when the simulator is costly relative to the cost of MMD estimation, or , b) when $\nu_c$ is large and the dimension $s$ is low.}
\vspace{-3mm}
\label{fig:flowchart}
\end{figure}

\vspace{-2mm}

\paragraph{Computational Cost.} The total computational cost of our method is the sum of (i) the cost of simulating from the model, which is $\mathcal{O}(m  C_{\mathrm{gen}})$, where $C_{\mathrm{gen}}$  is the cost of sampling one data point, and (ii) the cost of estimating MMD, which is $\mathcal{O}(m^2+mn+n^2)$ for the V-statistic and $\mathcal{O}(m^3+mn+n^2)$ for the OW estimator. Our method is hence slightly more expensive when $m$ is large. However, the cost of the simulator is often the computational bottleneck, sometimes taking up to tens or hundreds of CPU hours per run; see \citet{Behrens2015, Kirby2022}. As a result, proposing data efficient likelihood-free inference methods \cite{Beaumont2009, Gutmann2016, Greenberg2019} is still an active research area. In cases where $\mathcal{O}(m C_{\mathrm{gen}}) \gg \mathcal{O}(m^3)$, the OW estimator is more efficient than the V-statistic as it requires fewer simulations to estimate the MMD. If the simulator is not more expensive than estimating the MMD and assuming a fixed computational budget, then the OW estimator achieves lower error than the V-statistic if $\nu_c / s > 1/4$ and assumptions \Cref{as:points,as:generator,as:kernels} hold. This result is straightforwardly derived from \Cref{thm:rate_of_convergence}, see \Cref{app:cost_error} for details. \Cref{fig:flowchart} summarises the cases in which one should opt for our OW estimator instead of the V-statistic estimator.

We remark that the cost of inverting the kernel matrix in our method (\Cref{eq:optimal_weights}) could be reduced by using specific pairs of kernel and point sets; see \citet{Jagadeeswaran2019,Karvonen2019-kk,Karvonen2019-mu}. In this case, significant gains could be observed for even cheaper simulators.

\begin{figure*}
     \centering
     \begin{subfigure}[b]{0.26\textwidth}
         \centering
         \includegraphics[width=\textwidth]{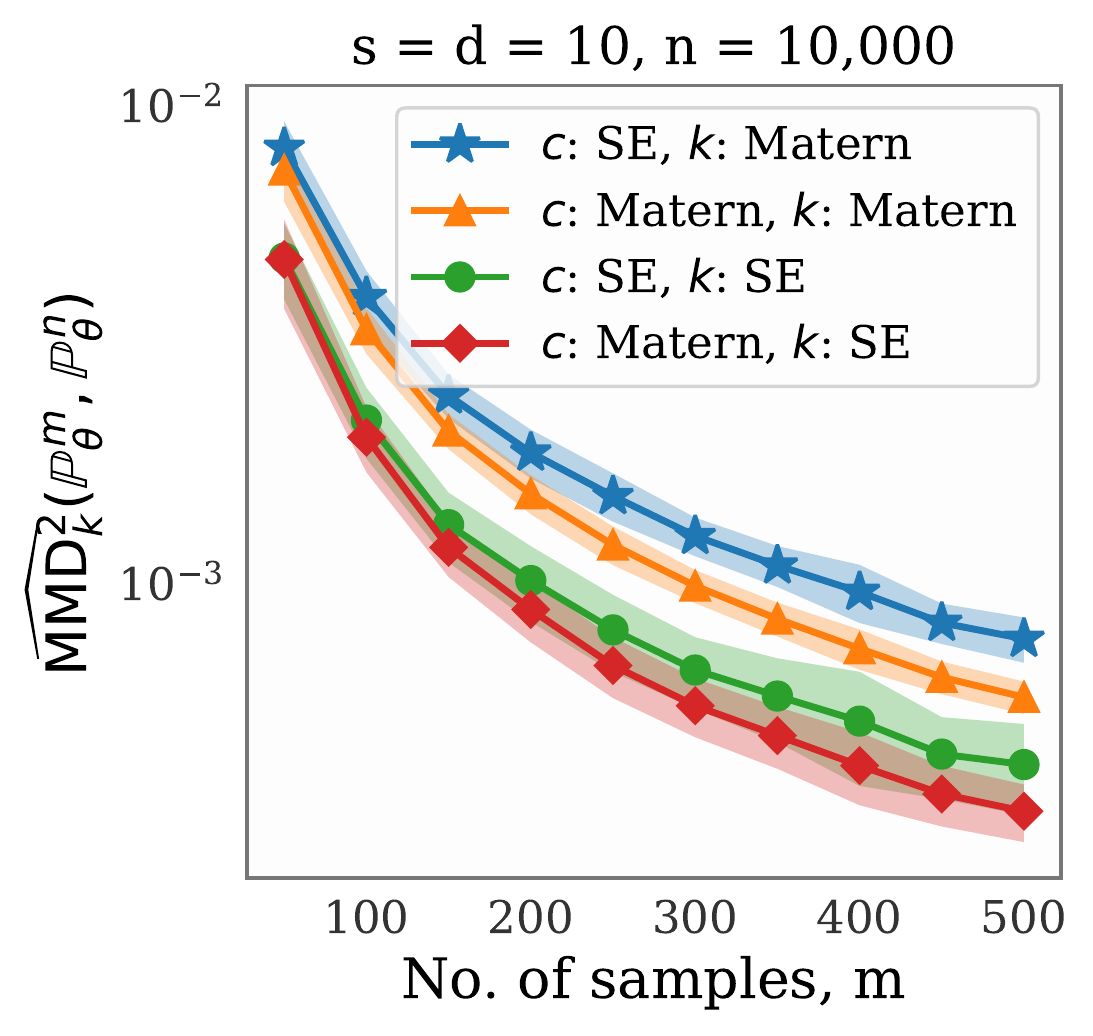}
         \vspace{-3ex}
         \caption{}
         \label{fig:mvgk_a}
     \end{subfigure}
     \begin{subfigure}[b]{0.24\textwidth}
         \centering
         \includegraphics[width=\textwidth]{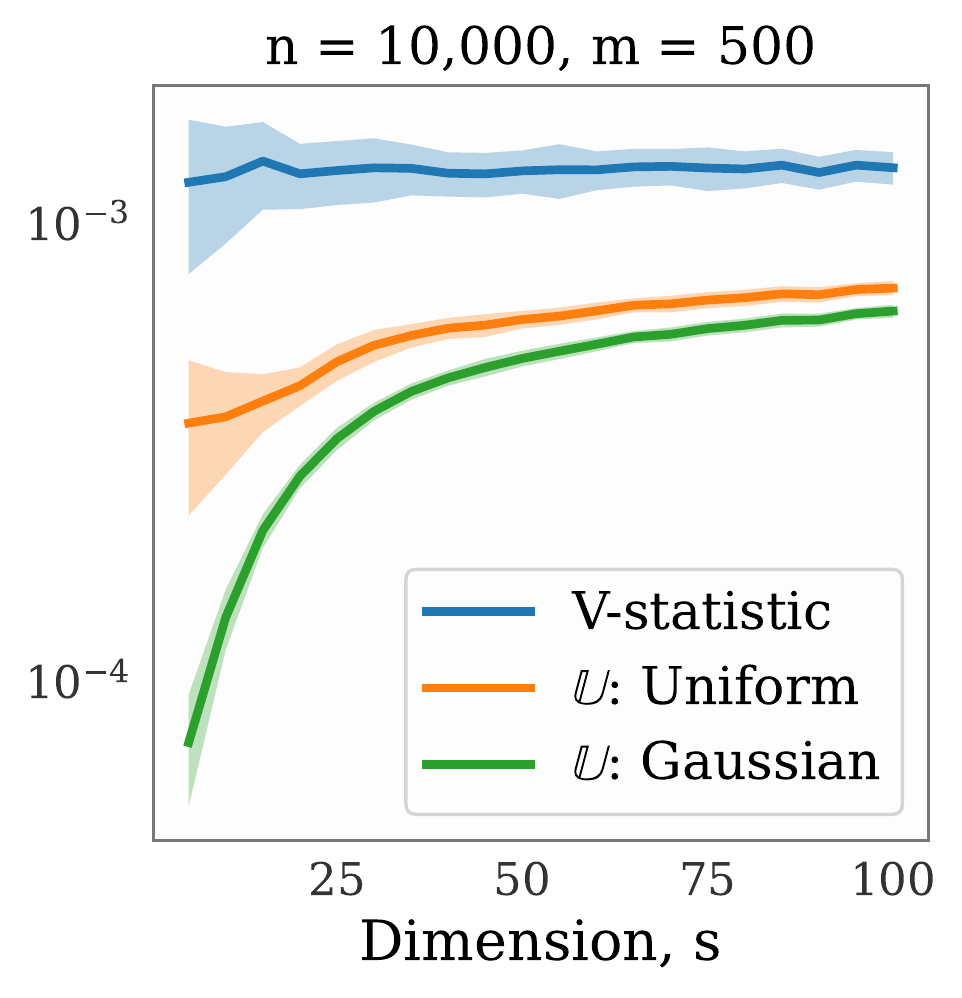}
         \vspace{-3ex}
         \caption{}
         \label{fig:mvgk_b}
     \end{subfigure}
     \begin{subfigure}[b]{0.24\textwidth}
         \centering
         \includegraphics[width=\textwidth]{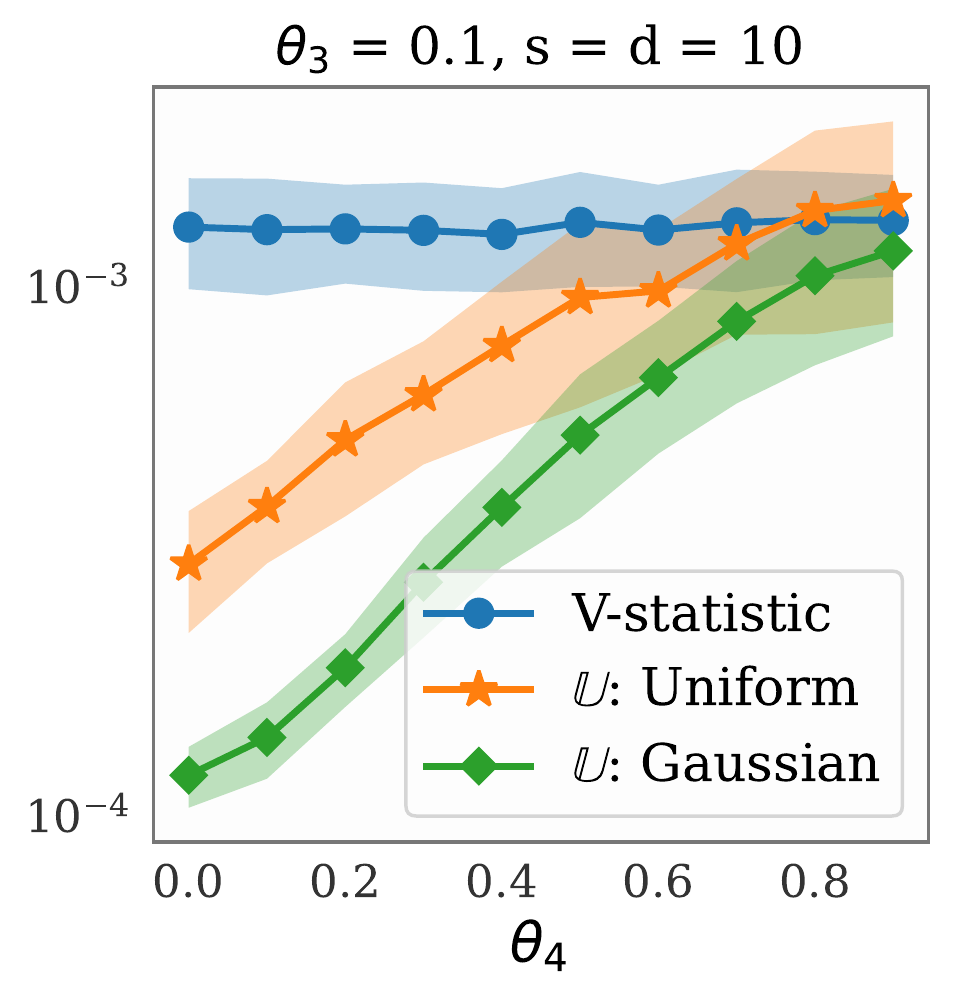}
         \vspace{-3ex}
         \caption{}
         \label{fig:mvgk_c}
     \end{subfigure}
     \begin{subfigure}[b]{0.243\textwidth}
         \centering
         \includegraphics[width=\textwidth]{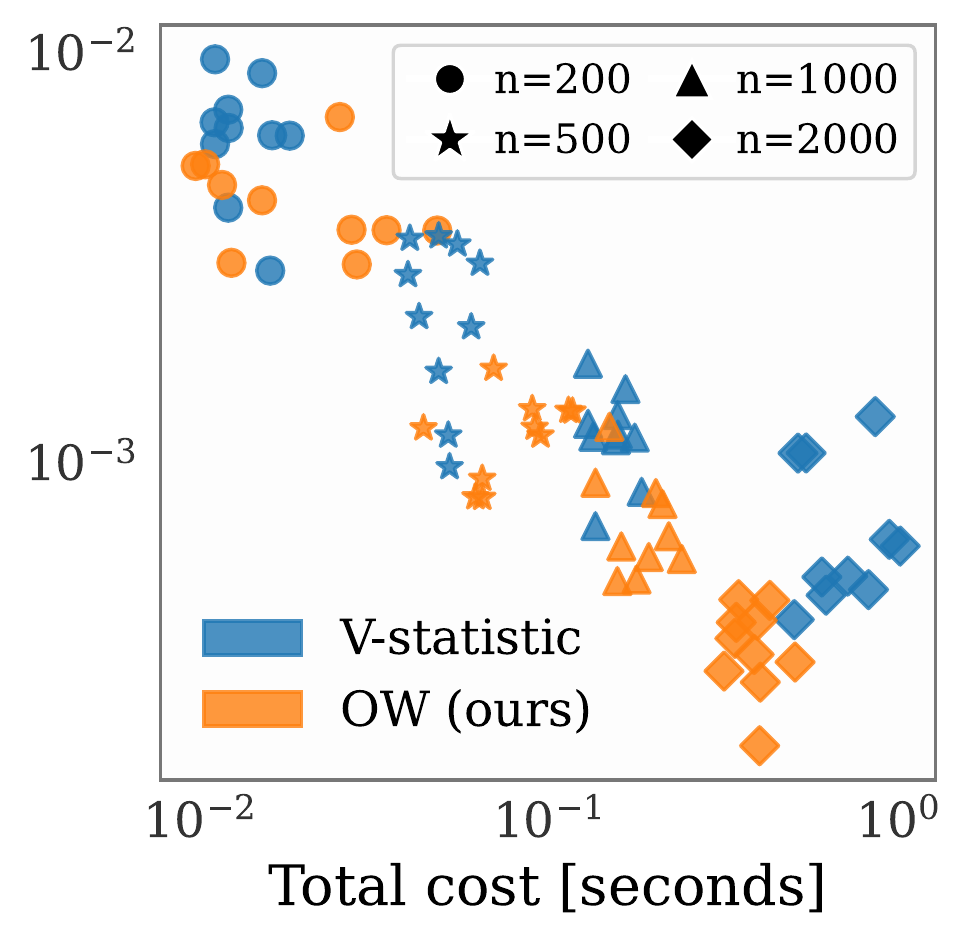}
         \vspace{-3ex}
         \caption{}
         \label{fig:mvgk_d}
     \end{subfigure}
        \caption{Error in estimating MMD$^2$ for the multivariate g-and-k distribution. (a) Error of our OW estimator for different choices of $k$ and $c$. Increasing the smoothness of $k$ improves the performance. (b) Comparison of V-statistic and OW estimator as a function of dimension. OW performs better for both parametrisations of $\mathbb{U}$, with the Gaussian giving lowest error. (c) Value of $\theta_4$ also impacts the performance of the OW estimator. (d) Error vs. total computation cost for different $n$. OW performs better than V-statistic for similar cost: $m=n$ for V-statistic, whereas $m = (68, 126, 200, 317)$ for OW.}
        \label{fig:mvgk_all}
\end{figure*}

\section{Numerical Experiments} \label{sec:experiments}

We now illustrate the performance of our OW estimator on various benchmark simulators and on challenging likelihood-free inference tasks. The lengthscale of kernels $k$ and $c$ is set using the median heuristic \cite{Garreau2018}, unless otherwise stated. The closed-form kernel mean embeddings used in the experiments are derived in \Cref{app:derivation_aligns}. Our code is available at \url{https://github.com/bharti-ayush/optimally-weighted_MMD}. 

\subsection{Benchmarking on popular simulators}\label{sec:benchmark}

We begin by comparing the V-statistic with our OW estimator on a number of popular benchmark simulators having different dimensions for $\mathcal{U} \subseteq \R^s$ and $\mathcal{X} \subseteq \R^d$. The experiments are conducted for $\{u_i\}_{i=1}^m$ being iid as well as RQMC points. We fix $\theta$ for each model (see \Cref{app:trueParam} for exact values) and estimate the MMD$^2$ between $\P_{\theta}^m$ and $\P_{\theta}^n $, with $k$ and $c$ both being the SE kernel. We set $n = 10,000$ to be large in order to make $\P_\theta^n$ an accurate approximation of $\P_\theta$, and $m=2^8$ so as to facilitate comparison with RQMC, which requires $m$ to be a power of $2$. 

The results are reported in \Cref{tab:MMDerror}. For RQMC points, the errors are generally either similar for the two estimators (g-and-k, two moons, and Lotka-Volterra models) or smaller for the OW estimator (bivariate Beta and MA$(2)$), with the OW estimator achieving lower errors in all cases barring the M/G/1 queuing model. This is not surprising since the M/G/1 model has a discontinuous generator, and our theory therefore does not hold. It is also important to note that although RQMC performs very well here even without the optimal weights, the simulators were chosen in order to make this comparison feasible. In many cases, $\mathbb{U}$ will not be uniform and therefore the RQMC approach will not be possible to implement and only the iid approach is feasible. 

For the iid points, the improvement in performance is much more significant. The OW estimator achieves the lowest error for all the models when  $\{u_i\}_{i=1}^m$ are taken to be iid uniforms. Its error is reduced by a factor of around $20$ and $40$ for the g-and-k and the two moons model, respectively, compared to the V-statistic. As expected from our sample complexity results, the magnitude of this improvement reduces as $s$ (the dimension of $\mathcal{U}$) increases. However, the OW estimator still performs slightly better than the V-statistic for the Lotka-Volterra model where $s=600$.

\subsection{Multivariate g-and-k distribution}
\label{sec:mvgk}

We now assess the impact of various practical choices on the performance of our method. To do so, we consider the multivariate extension of the g-and-k distribution introduced in \cite{Drovandi2011} and used as a benchmark in \cite{Li2017Copula, Jiang2018, Nguyen2020}. This flexible parametric family of distributions does not have a closed-form likelihood, but is easy to simulate from.
We define a distribution in this family through $ (G_\theta, \mathbb{U}_\theta)$, where
%
\begin{talign*}
     G_{\theta}(u)  =   \theta_1 \hspace{-0.5ex} +  \theta_2 \hspace{-0.5ex} \left[1 \hspace{-0.5ex} + \hspace{-0.5ex} 0.8 \frac{1 - \exp(-\theta_3 z(u))}{1 + \exp(-\theta_3 z(u))}\right] \hspace{-0.5ex} \left(1 \hspace{-0.5ex} + \hspace{-0.5ex} z(u)^2 \right)^{\theta_4} \hspace{-0.5ex} z(u),
\end{talign*}
with $\theta=(\theta_1,\theta_2,\theta_3,\theta_4,\theta_5)$, $z(u) = \Sigma^{\frac{1}{2}}u$ and
$\mathbb{U} = \mathcal{N}(0, I_s)$, where $\Sigma \in \mathbb{R}^{d \times d}$ is a symmetric tri-diagonal Toeplitz matrix such that $\Sigma_{ii} = 1$ and $\Sigma_{ij} = \theta_5$.
The parameters $\theta_1$,$\theta_2$,$\theta_3$, and $\theta_4$ govern the location, scale, skewness, and kurtosis respectively, and $s = d$.
An alternative formulation is through  $(\tilde{\mathbb{U}},\tilde{G}_\theta)$ where $\tilde{\mathbb{U}} = \text{Unif}(0,1)^s$, and $\tilde{G}_\theta = G_{\theta} \circ \Phi^{-1}$ where $\Phi$ is the cumulative distribution function of a $\mathcal{N}(0,1)$.

\vspace{-2mm}

\paragraph{Varying choice of $k$ and $c$.}
We first investigate the performance of our OW estimator for different combinations of $k$ and $c$, the choices being either the SE or the Mat\'ern kernel. We estimate the squared-MMD for each of these combinations as a function of $m$, with $d = 10$ and $n=10,000$. The Lebesgue measure formulation is used while computing the embeddings for both the kernels. The Mat\'ern kernel is set to order $\nu_k = \nu_c = 2.5$, and the parameter value to $\theta_0=(3,1,0.1,0.1,0.1)$. The resulting curves are shown in \Cref{fig:mvgk_a}. Our method performs best when $k$ is the SE kernel, i.e., when it is infinitely smooth. The performance degrades slightly when $k$ is Mat\'ern, while the combination of $c$ as SE and $k$ as the Mat\'ern kernel is the worst. This is expected from our theory, and is because the composition of $G_\theta$ and $k$ is not smooth, but we approximate it with an infinitely smooth function. Hence, from a computational viewpoint, it is always beneficial to take $k$ to be very smooth.

\begin{figure*}[t]
	 \centering
	 \includegraphics{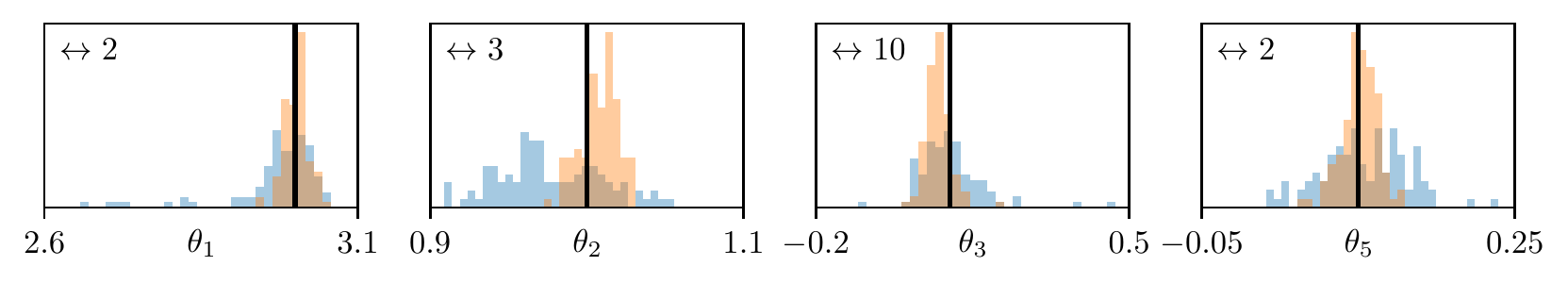}
  \vspace{-3mm}
\caption{
    Histogram of parameter estimates obtained using \Cref{eq:MMD_estimator} with the V-statistic estimator (blue) and the OW estimator (orange) over 100 runs during the composite goodness-of-fit test.
    The black vertical lines denote the true value of the parameter.
    $\leftrightarrow$ indicates the number of estimates for each parameter from the V-statistic that are outliers and hence not included in the plot.
    For the OW estimator, all estimates are within the x-axis range. The parameter estimates obtained using the OW estimator are more concentrated around the true parameter value, whereas the estimates obtained using the V-statistic have higher variance.
    }
\label{fig:gandk_min_dist}
\end{figure*}

\vspace{-2mm}

\paragraph{Varying dimensions $s$ and $d$.}
We now analyse the impact of the choice of measure, either Gaussian or uniform. \Cref{fig:mvgk_b} shows the OW and V-statistic estimators as  the dimension $s=d$ varies. The parameter values are the same as before, $m = 500$, and the SE kernel is used for both $k$ and $c$. We observe that the OW estimator performs better than the V-statistic even in dimensions as high as 100. In lower dimensions, the Gaussian embedding achieves lower error than the uniform for this model, with their performance converging around $d = 60$. This is likely due to the fact that $\tilde{G}_\theta$ is an easier function to approximate than $G_\theta$, but this is harder to assess a-priori for the user and highlights some open questions not yet covered by our theory.

\vspace{-2mm}

\paragraph{Varying model parameters.}
Building on the previous result, we show that the performance of the OW estimation is also impacted by $\theta$. In \Cref{fig:mvgk_c}, we analyse the performance of the estimators as a function of parameter $\theta_4$. The SE kernel is used for both $k$ and $c$. While the V-statistic is not impacted by the choice of $\theta_4$, the performance of our estimators degrade as $\theta_4$ increases. The behaviour is similar on varying $\theta_3$, albeit not as drastic as  $\theta_4$, see \Cref{app:mvgk} for the plot. We expect that this difference in performance is due to the regularity of the generator varying with $\theta$.

\vspace{-2mm}

\paragraph{Performance vs. computational cost.}
Finally, since the OW estimator tends to be more computationally expensive and this simulator is relatively cheap ($\approx$ 1~ms to generate one sample), we also compare estimators for a fixed computational budget. To that end, we vary $n$ and take $m=n$ for the V-statistic and $m = 2n^{2/3}$ for the OW estimator. \Cref{fig:mvgk_d} shows their performance with respect to their total computational cost, including the cost of simulating from the model ($d = s = 5$). We see that the OW estimator achieves lower error on average than the V-statistic. Hence, it is preferable to use the OW estimator even for a computationally cheap simulator like the multivariate g-and-k.

\vspace{-2mm}

\paragraph{Composite goodness-of-fit test.}
We demonstrate the performance of our method when applied to composite goodness-of-fit testing, using the method proposed by \citet{Key2021} with a test statistic based on the squared-MMD.
Given iid draws from some distribution $\Q$, the test considers whether  $\Q$ is an element of some parametric family $\{\P_{\theta}: \theta \in \Theta\}$ (null hypothesis) or not (alternative hypothesis).
The approach uses a parametric bootstrap \citep{Stute1993} to estimate the distribution of the squared-MMD under the null hypothesis, which can then be used to decide whether or not to reject. This requires repeatedly performing two steps: (i) estimating a parameter value through an MMD estimator of the form in \Cref{eq:MMD_estimator}, and (ii) estimating the squared-MMD between $\Q$ and the model at the estimated parameter value. See \Cref{app:gof} for the full algorithm. This needs to be done up to $B$ times, where $B$ can be in the hundreds or thousands, which can be a significant challenge computationally. This limits the number of simulated samples $m$ that can be used at each step, and is therefore a prime use case for our OW estimator.  

We performed this test with a level of $0.05$ using the V-statistic and OW estimator, using $B=200$. We considered the multivariate g-and-k model with unknown $\theta_1, \theta_2, \theta_3$, and $\theta_5$ but fixed $\theta_4 = 0.1$. We used $m=100$ and $n=500$ and considered two cases: $\Q$ is a multivariate g-and-k with $\theta_4=0.1$ (null holds) or $\theta_4=0.5$ (alternative holds). When the null hypothesis holds, we should expect the tests to reject the null at a rate close $0.05$, whereas when the alternative holds, we should reject at a rate close to $1$.
\Cref{tab:composite_test_results} shows that our test based on the OW estimator performs significantly better in that respect than the V-statistic. This is due to the fact that the OW estimator is able to improve both the estimate of the parameter (see \Cref{fig:gandk_min_dist}), and the estimate of the test statistic, thus improving the overall performance.

\begin{table}[]
    \centering
    \caption{Fraction of repeats for which the null was rejected. An ideal test would have $0.05$ when the null holds, and $1$ otherwise. }
    \vspace{2mm}
    \resizebox{0.9\columnwidth}{!}{%
    \begin{tabular}{c c c}
        \toprule
        Cases & IID V-stat & IID OW (ours)  \\ \midrule
         $\theta_4 =0.1$ (null holds) & 0.040 & 0.047 \\
         $\theta_4 =0.5$ (alternative holds) & 0.040 & 0.413 \\
        \bottomrule
    \end{tabular}
    }
    \label{tab:composite_test_results}
\vspace{-4mm}
\end{table}

\Cref{fig:gandk_min_dist} shows that the estimates of the parameters computed using the OW estimator are more concentrated around the true parameter value, whereas the estimates computed using the V-statistic have higher variance.
Therefore, when using the V-statistic, the distribution of the test statistic approximated by the bootstrap has higher variance, thus the estimated critical value is more conservative, and the test is not sensitive to smaller departures from the null hypothesis.
In contrast, when using the OW estimator, the estimated critical value is less conservative and the test has higher performance.

\subsection{Large scale offshore wind farm model}

Finally, we consider a low-order wake model \cite{Niayifar2016, Kirby2023} for large-scale offshore wind farms. The model simulates an estimate of the farm-averaged local turbine thrust coefficient \cite{Nishino2016}, which is an indicator of the energy produced. The parameter $\theta$ is the angle (in degrees) at which the wind is blowing. The turbulence intensity is assumed to have zero-mean additive Gaussian noise (i.e. $\mathbb{U} = \mathcal{N}(0,10^{-3})$), which then goes through the non-linear mapping of the generator. Although this model is an approximation of the state-of-the-art models that can take around 100 CPU hours per run (see e.g. \cite{Kirby2022}), one realisation from this model takes $\approx 2$~mins, which is still computationally prohibitive for likelihood-free inference. 
This example is indicative of the expensive simulators which are widely used in science, and is thus suitable for our method. 

We apply the ABC method of \eqref{eq:ABC} to estimate $\theta$ with both the OW estimator and the V-statistic. The tolerance threshold $\varepsilon$ is taken in terms of percentile, i.e., the proportion of the data that yields the least MMD distances. We use $1000$ parameter values from the $\text{Unif}(0,30)$ prior on $\Theta$. As the cost of the model far exceeds that of estimating the MMD, we take $m=10$ for both estimators. With few $m$, setting the lengthscale of $c$ using median heuristic is difficult, so we fix it to be 1.  The simulated datasets took $\approx 245$~hours to generate, while estimating the MMD took around $0.13$~s and $0.36$~s for the V-statistic and the OW estimator, respectively. 

The resulting posteriors, which are approximations of the ABC posterior obtained if the MMD was computable in closed-form, are in \Cref{fig:windfarm}. We observe that the OW estimator's posterior is much more concentrated around the true value than that of the V-statistic for both values of $\varepsilon$. This is because the OW estimator approximates the MMD more accurately than the V-statistic for the same $m$. Hence, our method can achieve similar performance as the V-statistic with much smaller $m$, saving hours of computation time. 

\begin{figure}[t!]
	 \centering
	 \includegraphics[ width = 0.43 \linewidth]{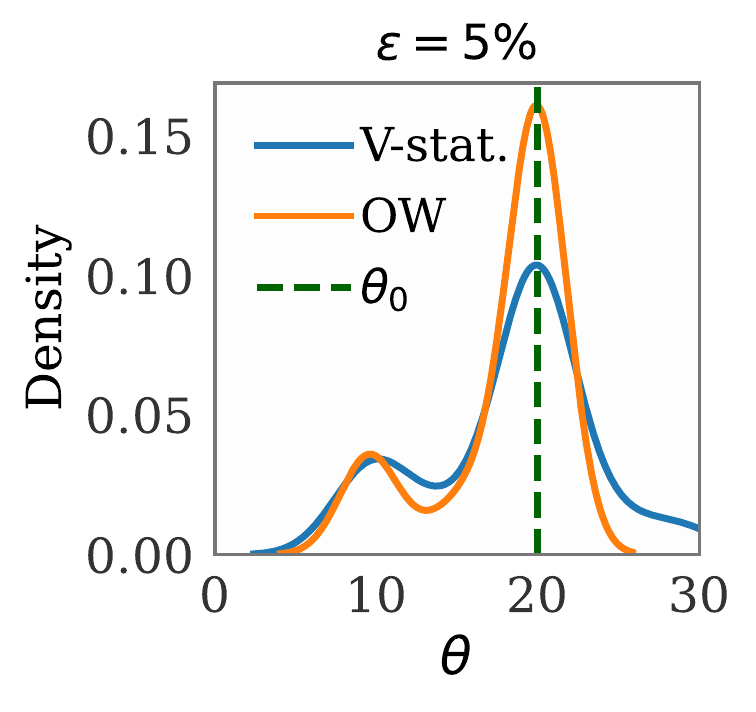}
	 \includegraphics[ width = 0.395 \linewidth]{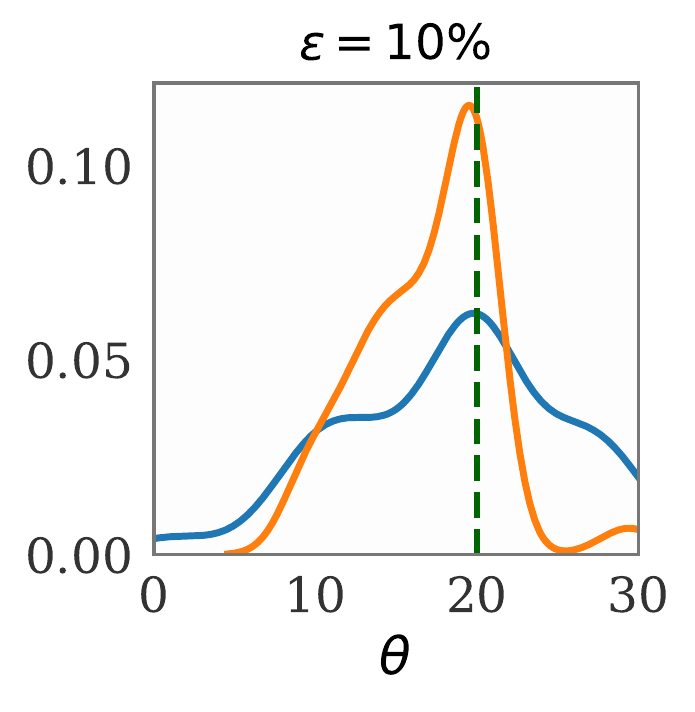}
  \vspace{-2mm}
\caption{ABC posteriors for the wind farm model. Our OW estimator yields posterior samples that are more concentrated around the true $\theta_0$ than the V-statistic. Performance of the U-statistic estimator is similar to the V-statistic, see \Cref{app:windfarm}. Settings: $n=100$, $\theta_0=20$.}
  \vspace{-3mm}
\label{fig:windfarm}
\end{figure}

\section{Conclusion} \label{sec:conclusion}

We proposed an optimally-weighted MMD estimator which has improved sample complexity than the V-statistic when the generator and kernel are smooth and the dimensionality is small or moderate. Thus, our estimator requires fewer data points than alternatives in this setting, making it especially advantageous for computationally expensive simulators which are widely used in the natural sciences, biology and engineering. However, a number of open questions remain, and we highlight the most relevant below. 

The parameterisation of a simulator through a generator $G_\theta$ and a measure $\U$ is usually not unique, and it is often unclear which parameterisation is most amenable to our method. One approach would be to choose a parameterisation where the dimension of $\U$ is small so as to improve the convergence rate. However, our result in \Cref{thm:rate_of_convergence} also contains rate constants which are difficult to get a handle on, and it is therefore difficult to identify which parameterisation is best amongst those with fixed smoothness and dimensionality.

Finally, our sample complexity result could be extended. One limitation is that we focus on the MMD and not its gradient, meaning that our results are not directly applicable for gradient-based likelihood-free inference such as the method used for our g-and-k example \cite{Briol2019MMD}. A future line of work could also investigate if our ideas translate to other distances used for likelihood-free inference, such as the Wasserstein distance \cite{Bernton2019} and Sinkhorn divergence \cite{Genevay2017,Genevay2018}.

\subsubsection*{Acknowledgements}

AB was supported by the Academy of Finland (Flagship programme: Finnish Center for Artificial Intelligence FCAI).
MN and OK acknowledge support from UKRI under the EPSRC grant
number [EP/S021566/1]. MN was also supported through The Alan Turing Institute’s Enrichment Scheme. SK was supported by the UKRI Turing AI World-Leading Researcher Fellowship, [EP/W002973/1].
FXB was supported by the Lloyd’s Register Foundation Programme on Data-Centric Engineering and The Alan Turing Institute under the EPSRC grant [EP/N510129/1], and through an Amazon Research Award on “Transfer Learning for Numerical Integration in Expensive Machine Learning Systems”.

\bibliography{bibliography}
\bibliographystyle{icml2023}

\newpage
\appendix

\onecolumn

{
\begin{center}
\Large
    \textbf{Supplementary Materials}
\end{center}
}
In \Cref{app:proofs}, we present the proofs and derivations of all the theoretical results in our paper, while \Cref{app:expDetails} contains additional details regarding our experiments.

\section{Proof of Theoretical Results}\label{app:proofs}

In this section, we prove~\Cref{thm:optimal_weights,thm:rate_of_convergence} and intermediate results required, and expand on the technical background. 

\subsection{Proof of \Cref{thm:optimal_weights}}\label{app:proofs_optimal_weights}
\begin{proof}

Let $\P_\theta^{m,w} = \sum_{i=1}^m w_i \delta_{y_i} = \sum_{i=1}^m w_i \delta_{G_\theta(u_i)}$. Using the fact that the MMD is a metric, we can use the reverse triangle inequality to get
\begin{talign*}
 &\left| \text{MMD}_k(\P_\theta,\Q) - \text{MMD}_k(\P_\theta^{m,w},\Q)\right|
 \leq  \text{MMD}_k(\P_\theta,\P_\theta^{m,w}). 
\end{talign*}
Define a kernel $c_\theta$ on $\calU$ as $c_\theta(u, u') = k(G_\theta(u), G_\theta(u'))$. As $\P_\theta$ is a pushforward of $\U$ under $G_\theta$, it holds that:
\begin{talign*}
 \text{MMD}^2_k(\P_\theta,\P_\theta^{m,w}) &= \int_{\X} \int_{\X}  k(x,x') \P_\theta (\d x)  \P_\theta (\d x') - 2 \sum_{i=1}^m w_i \int_{\X}  k(x_i,x) \P_\theta (\d x) + \sum_{i,j=1}^m w_i w_j k(x_i,x_j)\\
 &= \int_{\X} \int_{\calU}  k(G_\theta(u),G_\theta(u')) \U (\d u)  \U (\d u') - 2 \sum_{i=1}^m w_i \int_{\calU}  k(G_\theta(u_i), G_\theta(u)) \U (\d u) \\
 &\hspace{167pt} + \sum_{i,j=1}^m w_i w_j k(G_\theta(u_i), G_\theta(u_j))\\
 & = \text{MMD}^2_{c_\theta}(\U,\sum_{i=1}^m w_i \delta_{u_i}). 
\end{talign*}
Since $c_\theta(u, \cdot) \in \Hc$ for all $u \in \calU$---by the assumption that $k(x, \cdot) \circ G_\theta \in \Hc$ for all $x \in \X$---it holds that $\mathcal{H}_{c_\theta} \subseteq \mathcal{H}_c$. If $\mathcal{H}_{c_\theta} = \mathcal{H}_c$, we have $\text{MMD}_k(\P_\theta,\P_\theta^{m,w})=\text{MMD}_c(\U,\sum_{i=1}^m w_i \delta_{u_i})$, and the result holds for $K=1$.

Suppose $\mathcal{H}_{c_\theta} \subset \mathcal{H}_c$. Then, by~\citet[Theorem I.13.IV]{aronszajn1950theory}, for any $f \in \mathcal{H}_{c_\theta}$ there is a constant $K$ independent of $f$ such that $\|f\|_{\Hc} \leq K \|f\|_{\Hct}$. Together with the fact that  $\text{MMD}_{c_\theta}$ is an integral-probability metric with underlying function class being the unit-ball in $\Hct$, this gives
\begin{talign*}
    \text{MMD}_{c_\theta}(\U,\sum_{i=1}^m w_i \delta_{u_i}) &= \sup_{\|f\|_{\Hct}\leq 1} \left| \int_{\calU} f(u) \U(\text{d}u) - \sum_{i=1}^m w_i  f(u_i) \right| \\
    &=K \times \sup_{\|f\|_{\Hct}\leq 1/K} \left| \int_{\calU} f(u) \U(\text{d}u) - \sum_{i=1}^m w_i  f(u_i) \right| \\
    &\leq K \times\sup_{\substack{f \in \Hct\\\|f\|_{\Hc}\leq 1}} \left| \int_{\calU} f(u) \U(\text{d}u) - \sum_{i=1}^m w_i  f(u_i) \right| \\
    &\leq K \times \sup_{\|f\|_{\Hc}\leq 1} \left| \int_{\calU} f(u) \U(\text{d}u) - \sum_{i=1}^m w_i  f(u_i) \right| \\
    &= K \times \text{MMD}_c(\U,\sum_{i=1}^m w_i \delta_{u_i}),
\end{talign*}
where the second equality is simply a reparametrisation from $f$ to $Kf$, and the inequalities use the fact that supremum of a set is not greater than supremum of its superset, and
\begin{talign*}
    \{f \in \Hct \ |\ K \|f\|_{\Hct} \leq 1 \} \subseteq \{f \in \Hct\ |\ \|f\|_{\Hc} \leq 1 \} \subseteq \{f \in \Hc\ |\ \|f\|_{\Hc} \leq 1 \}.
\end{talign*}
Note that the tightness of the bound will depend on the gap between $\mathcal{H}_{c_\theta}$ and $ \mathcal{H}_c$; the smaller this gap, the tighter the bound will be. This is illustrated in~\Cref{fig:the_gap}.

To prove the result about the exact form of $w$, we note that 
\begin{talign*}
 \underset{ w \in \R^m}{\arg \min } ~ \text{MMD}_c\left(\U,\sum_{i=1}^m w_i \delta_{u_i}\right)= \underset{ w \in \R^m}{\arg \min }~ \text{MMD}^2_c\left(\U,\sum_{i=1}^m w_i \delta_{u_i}\right),
\end{talign*}
and 
\begin{talign*}
\text{MMD}^2_c\left(\U,\sum_{i=1}^m w_i \delta_{u_i}\right)  = \int_{\calU} \int_{\calU}  c(u,v) \U(\text{d}u) \U(\text{d}v) - 2 \sum_{i=1}^m w_i \int_{\calU}  c(u_i,u) \U(\text{d}u) + \sum_{i,j=1}^m w_i w_j c(u_i,u_j).
\end{talign*}
The latter is a quadratic form in $w$, meaning it can be minimised in closed-form over $w$ and the optimal weights are given by $w^*$. This completes the proof of the second part of the theorem.

\end{proof}

\begin{figure}
    \centering
    \includegraphics[width=0.2\textwidth]{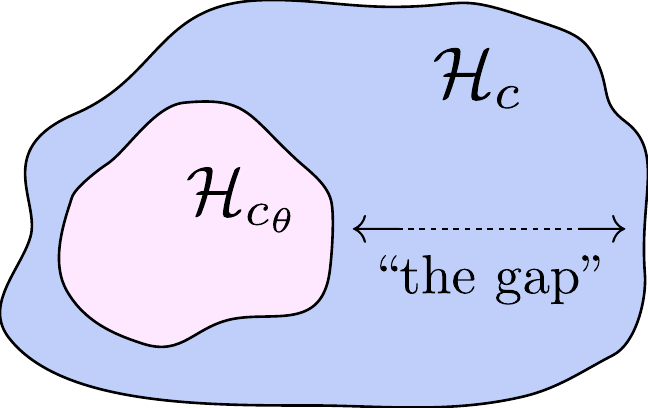}
    \caption{Pictorial representation of the gap between the RKHS induced by the kernels $c$ and $c_\theta = k(G_\theta(u), G_\theta(v))$. The size of the gap affects the tightness of the bound in~\Cref{thm:optimal_weights}, and consequently~\Cref{thm:rate_of_convergence}.} \label{fig:the_gap}
\end{figure}

\subsection{Chain rule in Sobolev spaces}

The proof of~\Cref{thm:rate_of_convergence}, specifically the result $k(x, \cdot) \circ G_\theta \in \Hc$ for Mat\'ern $k$ and $c$, will use a specific form of a chain rule for Sobolev spaces. We justify the choice of Mat\'ern kernels---or more generally, kernels the RKHS of which is norm-equivalent to the well-studied Sobolev space---and prove the form of the chain rule for Sobolev spaces that will imply $k(x, \cdot) \circ G_\theta \in \Hc$.

For general $c$ and $k$, $k(x, \cdot) \circ G_\theta \in \Hc$ is non-trivial to check. Here, we introduce sufficient conditions on $c$, $k$, and $G_\theta$ that are easily interpretable and correspond to common practical settings. Specifically, we consider $c$ and $k$ the RKHS of which, $\Hc$ and $\Hk$, are Sobolev spaces, and $G_\theta$ of a certain degree of smoothness---which reduces the problem to a form of a chain rule for Sobolev spaces.\footnote{Though various forms of the chain rule for Sobolev spaces exist in the literature (for example,~\citet[Section 4.2.2]{evans2018measure}), they tend to either consider $F \circ f$, where $f$ is in the Sobolev space (rather than $F$), or place overly strong assumptions on $f$.} The rest of the section proceeds as follows: first, we introduce the background definitions and results, then show that the required form of the chain rule holds for first order derivatives (\Cref{lemma:chain_rule_w1p}), and finally extend the result to higher order derivatives (\Cref{thm:k_circ_G_is_Sobolev}).

\paragraph{Background.} We consider the well-studied Sobolev kernels~\citep[see e.g.][Chapter 10]{Wendland2005}, which are kernels that induce a reproducing kernel Hilbert space (RKHS) that is norm-equivalent to a Sobolev space $\W^{l, 2}(\X)$, $\X \subseteq \R^d$, for some integer $l>d/2$. 
We give the definition of $\W^{l, 2}$ Sobolev spaces below, and refer to~\citet{adams2003sobolev} for an in-depth treatment of Sobolev spaces and~\citet{Berlinet2004} for general RKHS theory.

\begin{definition}[Sobolev spaces]
    Suppose $\X$ is an open subset of $\R^d$. The Sobolev space $\W^{l, 2}(\X)$, $l>d/2$, is a space of functions $f: \X \to \R$ such that $\|f\|^2_{\L^2(\X)} = \int_\X f^2(x) \d x < \infty$,
    and for any multi-index $\alpha \in \mathbb{N}^d$ with $|\alpha|=\sum_{i=1}^d \alpha_i \leq l$, the weak derivative $D^\alpha f = D^{\alpha_1}_{x_1} \dots D^{\alpha_d}_{x_d} f$ exists and $ \|D^\alpha f\|_{\L^2(\X)} < \infty$.
\end{definition}

A weak derivative is a generalisation of the concept of a derivative to functions that are not differentiable. A locally integrable function $D_{x_i} f$ is a weak derivative of $f$ in $x_i$ if it closely resembles the behavior of the ordinary derivative on any open $U \subseteq \X$: for any infinitely continuously differentiable function 
with a compact support, the integration chain rule holds with $f$ and $D_{x_i} f$---as it would for an ordinary derivative.
As the definition is only concerned with equality of the integrals in the chain rule,
a weak derivative is not uniquely defined: two functions $g_1$ and $g_2$ can be weak derivatives of $f$ in $x_i$ if (and only if) they only differ on a zero-volume set, meaning a set the Lebesgue measure of which is zero. 
As such, by $D_{x_i} f$ we will refer to any function that satisfies the definition of a weak derivative. For a multi-index $\alpha = (\alpha_1, \dots \alpha_d) \in \mathbb{N}^d$, by $D^\alpha f$ we denote the $|\alpha|$ order weak derivative $D^\alpha f=D^{\alpha_1}_{x_1} \dots D^{\alpha_d}_{x_d} f$, where
\begin{talign*}
D^n x_i f=\underbrace{D_{x_i} \dots D_{x_i}}_n f \text{ for any } n \in \mathbb{N}.
\end{talign*}
If an ordinary derivative $\partial^\alpha f=\partial^{|\alpha|} f/\partial^{\alpha_1} x_1 \dots \partial^{\alpha_d} x_d$ exists, it is equal to any weak $D^\alpha f$. It is important to clarify that the definition of Sobolev spaces given here is specific to the case $\W^{l, 2}(\X)$, $l>d/2$. General Sobolev spaces $W^{l, p}(\X)$ are subspaces of more general Lebesgue spaces, and are spaces not of functions, but of \emph{equivalence classes} of functions. Two functions $f_1$, $f_2$ are in the same equivalence class $[f]$ if they are equal almost everywhere. General Lebesgue and Sobolev space theory requires careful handling of the notion of equivalence classes, as the functions in them may differ arbitrarily on sets of Lebesgue measure zero. However, by Sobolev embedding theorem~\citep[Theorem 4.12]{adams2003sobolev} every element of $\W^{l, 2}(\X)$ is continuous if $l > d/2$, which implies that every equivalence class contains exactly one function---and we may define $\W^{l, 2}(\X)$ as a space of functions, as is done above.

Throughout the proofs, we will say $f \in \L^\infty(\X)$ if it is bounded on $\X$, and $f \in C^{m}(\X)$, for $m \in \mathbb{N}$, if $\partial^\alpha f$ exists and is continuous for any $|\alpha| \in [0, m]$. Specifically, $C^0(\X)$ is the space of continuous functions, and $C^\infty(\X)$ a space of infinitely differentiable functions with continuous derivatives. 
The output space of functions in both $\L^\infty(\X)$ and $C^{m}(\X)$ is omitted from the notation as it will be clear from the specific $f$ in question.

We start by recalling an important result that characterises Sobolev functions as limit points of sequences of $C^\infty(\X)$ functions.
Since it is a necessary and sufficient condition, we will use this result both to operate on a function in a Sobolev space using the "friendlier" smooth functions, and to prove a function of interest lies in a Sobolev space by finding a sequence of smooth function that approximates it accordingly.

\begin{theorem}[Theorem 3.17, \cite{adams2003sobolev}]
\label{thm:sobolev_function_approximation}
For an open set $\X \subseteq \R^d$, a function $f: \X \to \R$ lies in the Sobolev space $\W^{1, 2}(\X)$ and has weak derivatives $D_{x_j}[f]$, $j \in [1, d]$ %
if and only if there exists a sequence of functions $f_n \in C^\infty(\X) \cap \W^{1, 2}(\X)$ such that for $j \in [1, d]$
\begin{talign}
\label{eq:sobolev_function_approximation_funcs}
    \| f - f_n \|_{\L^2(\X)} \to 0,&\quad n \to \infty,  \\
    \left\| D_{x_j}[f] - \frac{\partial f_n}{\partial x_j} \right\|_{\L^2(\X)} \to 0,&\quad n \to \infty, \label{eq:sobolev_function_approximation_funcs2}
\end{talign}
where $\frac{\partial f_n}{\partial x_j}$ is the ordinary derivative of $f_n$ with respect to $x_j$.
\end{theorem}

Note that the functions $f_n$ converge to $f$ in the Sobolev $\W^{1, 2}(\X)$ norm, $\| f - f_n \|^2_{\W^{1, 2}(\X)} = \|f -f_n\|_{\L^2(\X)} + \sum_{j=1}^d \|D_{x_j} f- \partial f_n / \partial x_j \|_{\L^2(\X)} \to 0$ as $n \to \infty$ , if and only if ~\eqref{eq:sobolev_function_approximation_funcs} and \eqref{eq:sobolev_function_approximation_funcs2} hold.

\paragraph{Chain rule for $\W^{1, 2}$.} We now prove that chain rule holds for $\varphi \circ G_\theta$ for $\varphi$ in a Sobolev space $W^{1, 2}(\X)$. For clarity, we will explicitly state the assumptions on $G_\theta$ in the main text. Recall that a measure $\P_\theta$ on $\X \subseteq \R^d$ is said to be a pushforward of a measure $\U$ on $\calU \subseteq \R^s$ under $G_\theta: \calU \to \X$ if for any $\X$-measurable $f: \X \to \R$ it holds that $\int_\X f(x) \P_\theta(\d x)= \int_\calU \left[f \circ G_\theta\right](u) \U(\d u)$.

\begin{lemma}[Chain rule for $\W^{1, 2}$]
\label{lemma:chain_rule_w1p}

Suppose
\begin{itemize}
  \item $\varphi \in \W^{1, 2}(\X)$. 
  \item $\calU \subset \R^{s}$ is bounded, $\X \subset \R^d$ is open, and $\X = G_\theta(\calU)$ for some $G_\theta = (G_{\theta, 1}, \dots, G_{\theta, d})^\top$. %
  The partial derivative $\partial G_{\theta, j}/\partial u_i$ exists and $|\partial G_{\theta, j}/\partial u_i| \leq C_G$ for some $C_G$ for all $i \in [1, s]$ and $j \in [1, d]$.
  \item $\U$ is a probability distribution on $\calU$ that has a density $f_\U: \calU \to [C_\U, \infty)$ for $C_\U>0$. 
  \item $\P_\theta$ is a pushforward of $\U$ under $G_\theta$, and has a density $f_{\P_\theta}$ such that $f_{\P_\theta}(x) \leq C_{\P_\theta}$ for all $x \in \X$ for some $C_{\P_\theta}$.
\end{itemize}
Then $\varphi \circ G_\theta \in \W^{1, 2}(\calU)$, and for $i \in [1, s]$, its weak derivative $D_{u_i}[\varphi \circ G_\theta]$ is equal to $\sum_{j=1}^d [D_{x_j} \varphi \circ G_\theta] \frac{\partial {G_{\theta, j}}}{\partial u_i}$.
\end{lemma}

\begin{proof}

Since $\X$ is open, by \Cref{thm:sobolev_function_approximation} there is a sequence $\varphi_n \in C^\infty(\X) \cap \W^{1, 2}(\X)$ such that
\begin{talign*}
    \| \varphi - \varphi_n \|_{\L^2(\X)} \to 0,&\quad n \to \infty, \\
    \left\| D_{x_j}\varphi - \frac{\partial \varphi_n}{\partial x_j} \right\|_{\L^2(\X)} \to 0,&\quad n \to \infty,
\end{talign*}
The proof proceeds as follows: we show that the sequence $\varphi_n \circ G_\theta$ approximates $\varphi \circ G_\theta$, and $\frac{\partial \left[\varphi_n \circ G_\theta\right]}{\partial u_i}$ approximates the sum in the statement of the lemma, $\sum_{j=1}^d [D_{x_j} \varphi \circ G_\theta] \frac{\partial {G_{\theta, j}}}{\partial u_i}$, in $\L^2(\calU)$--norm. Then, by the sufficient condition in \Cref{thm:sobolev_function_approximation}, $\varphi \circ G_\theta$ lies in $\W^{1, 2}(\calU)$, and its weak derivative in $u_i$ is $\sum_{j=1}^d [D_{x_j} \varphi \circ G_\theta](u) \frac{\partial {G_{\theta, j}}}{\partial u_i}(u)$, for any $i \in [1, s]$.

Since $\P_\theta$ has a density,
for any $\X$-measurable $f$ it holds that
\begin{talign*}
    \int_\X f(x) f_{\P_\theta}(x) \d x = \int_\calU \left[f \circ G_\theta\right](u) f_\U(u) \d u.
\end{talign*}
Together with density bounds, this gives $\|\varphi \circ G_\theta - \varphi_n \circ G_\theta\|_{\L^2(\calU)} \to 0$ as
\begin{talign*}
   \int_\calU \left(\varphi \circ G_\theta(u) - \varphi_n \circ G_\theta(u) \right)^2 \d u 
   \leq C_\U^{-1} \int_\calU \left(\varphi \circ G_\theta(u) - \varphi_n \circ G_\theta(u) \right)^2 f_\U(u) \d u 
   &= C_\U^{-1} \int_\X \left(\varphi(x)  - \varphi_n(x)\right)^2 f_{\P_\theta}(x) \d x \\
   &\leq C_\U^{-1} C_{\P_\theta} \int_\X \left(\varphi(x)  - \varphi_n(x)\right)^2 \d x.
\end{talign*}
In the same fashion, $\|D_{x_j}\varphi \circ G_\theta - \frac{\partial \varphi_n}{\partial x_j} \circ G_\theta\|_{\L^2(\calU)} \to 0$ since
\begin{talign*}
   \int_\calU \big(D_{x_j}\varphi \circ G_\theta(u) - \frac{\partial \varphi_n}{\partial x_j} \circ G_\theta (u)\big)^2 \d u 
   & \leq C_\U^{-1} \int_\calU \big(D_{x_j}\varphi \circ G_\theta (u) - \frac{\partial \varphi_n}{\partial x_j} \circ G_\theta (u) \big)^2 f_\U(u) \d u \\
   &= C_\U^{-1} \int_\X \big(D_{x_j}\varphi(x)  - \frac{\partial \varphi_n}{\partial x_j} (x)\big)^2 f_{\P_\theta}(x) \d x \nonumber  \\
   &\leq C_\U^{-1} C_{\P_\theta} \int_\X \big(D_{x_j}\varphi (x)  - \frac{\partial \varphi_n}{\partial x_j} (x)\big)^2 \d x.
\end{talign*}
Since $\varphi$ and $G_\theta$ are both differentiable, the ordinary chain rules applies to $\varphi_n \circ G_\theta$,
\begin{talign*}
    \frac{\partial [\varphi_n \circ G_\theta]}{\partial u_i}= \sum_{j=1}^d \big[\frac{\partial \varphi_n }{\partial x_j} \circ G_\theta \big] \frac{\partial {G_{\theta, j}}}{\partial u_i},
\end{talign*}
and for any $i \in [1, s]$ the convergence of derivatives 
 $\|[D_{x_j} \varphi \circ G_\theta] \frac{\partial {G_{\theta, j}}}{\partial u_i} - \frac{\partial [\varphi_n \circ G_\theta]}{\partial u_i} \|_{\L^2(\calU)} \to 0$ follows since
\begin{talign}
    \int_\calU \Big( \sum_{j=1}^d \big[D_{x_j} \varphi \circ G_\theta\big] \frac{\partial {G_{\theta, j}}}{\partial u_i} - \frac{\partial [\varphi_n \circ G_\theta]}{\partial u_i} \Big)^2  \d u 
    &= \int_\calU \Big( \sum_{j=1}^d \big[D_{x_j}\varphi \circ G_\theta - \frac{\partial \varphi_n}{\partial x_j} \circ G_\theta\big]  \frac{\partial {G_{\theta, j}}}{\partial u_i} \Big)^2 \d u \nonumber \\
    &\leq 2 \sum_{j=1}^d \int_\calU \Big(  \big[D_{x_j}\varphi \circ G_\theta - \frac{\partial \varphi_n}{\partial x_j} \circ G_\theta\big]  \frac{\partial {G_{\theta, j}}}{\partial u_i} \Big)^2 \d u \nonumber \\
    &\leq 2 C_G^2 \sum_{j=1}^d \int_\calU \big(  D_{x_j}\varphi \circ G_\theta - \frac{\partial \varphi_n}{\partial x_j} \circ G_\theta  \big)^2 \d u \nonumber
\end{talign}
where
%
the first inequality is using the inequality $(\sum_{i=1}^d a_i)^2 \leq 2 \sum_{i=1}^d a_i^2$.
This completes the proof.

\end{proof}

\paragraph{Chain rule for $\W^{l, 2}$.} To extend~\Cref{lemma:chain_rule_w1p} to Sobolev spaces of order higher than 1, we need the following version of the weak derivative product rule, for a product of a function $f$ in $\W^{1, 2}$ and bounded differentiable function $g$ with bounded derivatives. Other versions of the product rule---for different regularity assumptions on $g$---exist in the literature (for example,~\citet{adams2003sobolev}); we will require this specific form.
\begin{lemma}[Product rule]
\label{lemma:product_rule}
    Suppose $\X \subseteq \R^d$ is open, $f \in \W^{1, 2}(\X)$, $g(x)$ is differentiable on $\X$, and $g(x) \leq L$, $[\partial g / \partial x_i ] (x) \leq L$ for all $x \in \X$ for some constant $L$. Then $fg \in \W^{1, 2}(\X)$ and for any $i \in [1, d]$,
    \begin{talign*}
        D_{x_i} [fg] = [ D_{x_i} f ] g + f \big[\partial g / \partial x_i\big]
    \end{talign*}
\end{lemma}
\begin{proof}
    By the criterion in~\Cref{thm:sobolev_function_approximation}, there is a sequence of smooth functions $f_n$ approximating $f$, meaning 
    \begin{talign*}
        \int_\X (f(x) - f_n(x))^2 \d x \to 0 \text{ as } n \to \infty,\\
        \int_\X \big(D_{x_i} f(x) - [\partial f_n/ \partial x_i](x)\big)^2 \d x \to 0 \text{ as } n \to \infty.
    \end{talign*}
     We will show that $f_n g$ approximates $fg$ with weak derivatives taking the form $[ D_{x_i} f ] g + f [\partial g / \partial x_i]$; by the aforementioned criterion, it will follow that $fg \in \W^{1, 2}(\X)$.

    First, we establish convergence of functions. As $n \to \infty$,
    \begin{talign*}
        \| fg - f_n g \|^2_{\L^2(\X)} = \int_\X \left(f(x) g(x) - f_n (x) g(x)\right)^2 \d x \leq L^2 \int_\X \left(f(x) - f_n (x) \right)^2 \d x \to 0.
    \end{talign*}
    By the ordinary chain rule, $\partial [f_n g] / \partial x_i = [\partial f_n / \partial x_i]g + f [\partial g / \partial x_i]$.
    Then, applying triangle inequality for norms and the fact that $(a+b)^2 \leq 2a^2 + 2b^2$ for any $a$, $b$, we get that for $n \to \infty$,
    \begin{talign*}
        \big\| \frac{\partial f_n }{ \partial x_i}g + f_n \frac{\partial g}{\partial x_i} - \left[ D_{x_i} f \right] g - f \frac{\partial g }{ \partial x_i} \big\|^2_{\L^2(\X)} 
        &\leq 2\left\| \frac{\partial f_n}{ \partial x_i}g - \left[ D_{x_i} f \right] g \right\|^2_{\L^2(\X)} + 2\left\| f_n \frac{\partial g}{ \partial x_i} - f \frac{\partial g}{ \partial x_i} \right\|^2_{\L^2(\X)} \\
        &\leq 2 L^2 \left\| \frac{\partial f_n}{ \partial x_i} - \left[ D_{x_i} f \right] \right\|^2_{\L^2(\X)} + 2 L^2 \left\| f_n - f \right\|_{\L^2(\X)} \to 0.
    \end{talign*}
    This completes the proof.
\end{proof}

We are now ready to extend the chain rule from order 1---proven in~\Cref{lemma:chain_rule_w1p}---to arbitrary order $l$.

\begin{theorem}[Chain rule for $\W^{l, 2}$]
\label{thm:k_circ_G_is_Sobolev}
Suppose
\begin{itemize}
  \item $\varphi \in \W^{l_\varphi, 2}(\X)$.
  \item $\calU \subset \R^{s}$ is bounded, $\X \subset \R^d$ is open, and $\X = G_\theta(\calU)$ for some $G_\theta = (G_{\theta, 1}, \dots, G_{\theta, d})^\top$. For some $l_G$ and any $|\alpha|\leq l_G$, $j \in [1,s]$, the derivative $\partial^\alpha G_{\theta,j}$ exists and is in $\L^\infty(\calU)$.
  \item $\U$ is a probability distribution on $\calU$ that has a density $f_\U: \calU \to [C_\U, \infty)$ for $C_\U>0$.
  \item $\P_\theta$ is a pushforward of $\U$ under $G_\theta$ with a density bounded above.
\end{itemize}
Then $\varphi \circ G_\theta \in \W^{l, 2}(\calU)$ for $l=\min\{l_\varphi, l_G\}$, and for any $k \leq l$ and $|\alpha_0| = k$, the derivative takes an $\alpha_0$-specific $(\kappa, \beta, \alpha, \eta)$--form
\begin{talign}
\label{eq:general_faa_di_bruno}
    D^{\alpha_0}[\varphi \circ G_\theta] = \sum_{i=1}^I \sum_{j=1 }^{d^{\kappa_i}} \left[ D^{\beta_{ij}} \varphi \circ G_\theta \right] \prod_{l=1}^{\kappa_i} \partial^{\alpha_{ijl}} G_{\theta, \eta_{ijl} },
\end{talign}
where $I \in \mathbb{N}$, and for any $i \in [1, I]$, $k \geq \kappa_i \in \mathbb{N}$; $\beta_{ij} \in \mathbb{N}^d$ is a multi-index of size $\kappa_i$ for $j \in [1, d^{\kappa_i}]$; $\alpha_{ijl} \in \mathbb{N}^s$ is of size $|\alpha_{ijl}| \leq k$, and $\eta_{ijl} \in [1, d]$ for $l \in [1, \kappa_i]$.
\end{theorem}

By saying the $(\kappa, \beta, \alpha, \eta)$ form is $\alpha_0$-specific, we mean that the values of $I, (\kappa, \beta, \alpha, \eta)$ depend on $\alpha_0$, and may be different for $\alpha'_0 \neq \alpha_0$; we do not index $I, (\kappa, \beta, \alpha, \eta)$ by $\alpha_0$ for the sake of readability.

Before proving this result, let us point out that the $(\kappa, \beta, \alpha, \eta)$--form introduced in the theorem can be seen as a form of the Faà di Bruno's formula which generalises the chain rule to higher derivatives~\citep[Theorem 1]{Constantine1996}. 
However, since our ultimate goal is to show $\varphi \circ G_\theta \in \W^{l, 2}(\calU)$, and the expression for the derivative is simply a means for proving that, an unspecified $(\kappa, \beta, \alpha, \eta)$--form suffices. It is simpler to conduct a proof for general $(\kappa, \beta, \alpha, \eta)$ without using explicit Faà di Bruno forms.

\begin{proof}[Proof of~\Cref{thm:k_circ_G_is_Sobolev}]
Note that $\varphi \circ G_\theta \in \W^{l, 2}(\calU)$ if and only if $\varphi \circ G_\theta \in \W^{k, 2}(\calU)$ for $k \leq l$. We use this to construct a proof by induction: we show the statement holds for $k=1$, and that $\varphi \circ G_\theta \in \W^{k, 2}(\calU)$ implies $\varphi \circ G_\theta \in \W^{k+1, 2}(\calU)$ if $k+1\leq l$ (and the weak derivatives take a $(\kappa, \beta, \alpha, \eta)$-form stated in~\Cref{eq:general_faa_di_bruno}).

\paragraph{Case $k=1$:} $\varphi \circ G_\theta$ is in $\W^{1, 2}(\calU)$. 
    
Suppose $\alpha_0=e[m]$ for some unit vector $e[m] = (0, \dots, 0, 1, 0, \dots, 0)$ where the $1$ is the $m$'th element. Then, as proven in~\Cref{lemma:chain_rule_w1p}, $D^{e[m]}[\varphi \circ G_\theta] = D_{u_m}[\varphi \circ G_\theta]$ is equal to $\sum_{j=1}^d [D_{x_j} \varphi \circ G_\theta] [\partial {G_{\theta, j}} / \partial u_m]=\sum_{j=1}^d [D^{e[j]} \varphi \circ G_\theta] \partial^{e[m]} {G_{\theta, j}} $, so the statement holds for $I=1$, $\kappa_1 = 1$, $\beta_{1j} = e[j]$, $\alpha_{1j1} = e[m]$, $\eta_{1j1} = j$.

\paragraph{Case $k$ implies $k+1$:} If $k+1 \leq l$ and $\varphi \circ G_\theta$ is in $\W^{k, 2}(\calU)$, and for every $|\alpha_0|=k$~\Cref{eq:general_faa_di_bruno} holds for some $\alpha_0$-specific $(\kappa, \beta, \alpha, \eta)$, then $\varphi \circ G_\theta$ is in $\W^{k+1, 2}(\calU)$, and for any $|\tilde{\alpha}_0|=k+1$ there is a  $(\tilde{\kappa}, \tilde{\beta}, \tilde{\alpha}, \tilde{\eta})$--form, $|\tilde{\kappa}|=\tilde I$,
\begin{talign}
\label{eq:induction_proof_eq5}
    D^{\tilde{\alpha}_0}[\varphi \circ G_\theta] = \sum_{i=1}^{\tilde{I}} \sum_{j=1 }^{d^{\tilde{\kappa}_i}} \left[ D^{\tilde{\beta}_{ij}} \varphi \circ G_\theta \right] \prod_{l=1}^{\tilde{\kappa}_i} \partial^{\tilde{\alpha}_{ijl}} G_{\theta, \tilde{\eta}_{ijl} }.
\end{talign}
By induction assumption, $\varphi \circ G_\theta$ is in $\W^{k, 2}(\calU)$, so it is in $\W^{k+1, 2}(\calU)$ if and only if $D^{\alpha_0} \left[ \varphi \circ G_\theta \right]$ is in $\W^{1, 2}(\calU)$ for any $\alpha_0$ of size $k$. The latter can be shown by studying the $(\kappa, \beta, \alpha, \eta)$--form that $D^{\alpha_0}[\varphi \circ G_\theta]$ takes by~\eqref{eq:general_faa_di_bruno}, for some $\alpha_0$--specific $(\kappa, \beta, \alpha, \eta)$. Since $l_\varphi \geq l \geq k+1$ (the last inequality holds by the induction assumption), it holds that $\W^{l_\varphi, 2}(\X) \subseteq \W^{l, 2}(\X) \subseteq \W^{k+1, 2}(\X)$. Then $\varphi \in \W^{k+1, 2}(\X)$, and since $|\beta_{ij}|=\kappa_i \leq k$ by definition of $\beta_{ij}$, we have $D^{\beta_{ij}} \varphi \in \W^{1, 2}(\X)$ for all $i, j$. Then by~\Cref{lemma:chain_rule_w1p}, its composition with $G_\theta$ is in $\W^{1, 2}(\calU)$, meaning $D^{\beta_{ij}}\varphi \circ G_\theta \in \W^{1, 2}(\calU)$. Consequently, $D^{\alpha_0}\left[\varphi \circ G_\theta\right]$ as per~\Cref{eq:general_faa_di_bruno} is a sum over the product of functions in $\W^{1, 2}(\calU)$, and bounded functions with bounded derivatives; by~\Cref{lemma:product_rule}, such product is in $\W^{1, 2}(\calU)$, and it follows that $D^{\alpha_0}\left[\varphi \circ G_\theta\right] \in \W^{1, 2}(\calU)$ as well.

Finally, we show that for any fixed $|\tilde{\alpha}_0|=k+1$ there are $\tilde{I}, \tilde{\kappa}, \tilde{\beta}, \tilde{\alpha}, \tilde{\eta}$ for which~\eqref{eq:induction_proof_eq5} holds; this will conclude the induction step.
Suppose $\alpha_0$ of size $k$, $|\alpha_0| = k$, is such that $\tilde{\alpha}_0 = \alpha_0 + e[m]$ for some $\alpha_0$ (that is unrelated to $\alpha_0$ in the previous part of the proof) and a unit vector $e[m]$ (such pair of $m$ and $\alpha_0$ must exist as $|\tilde{\alpha}_0|=k+1$). For this $\alpha_0$, in a slight abuse of notation, we shall say that $\kappa, \beta, \alpha, \eta$ are such that $D^{\alpha_0}[\varphi \circ G_\theta]$ takes a $(\kappa, \beta, \alpha, \eta)$ form. Then, by the sum rule for weak derivatives and the product rule of~\Cref{lemma:product_rule}, $D^{\tilde{\alpha}_0}\left[\varphi \circ G_\theta\right] = D_{u_m}\left[D^{\alpha_0}\left[\varphi \circ G_\theta\right]\right]$ takes the form
\begin{talign}
        D^{\tilde{\alpha}_0}\left[\varphi \circ G_\theta\right] = D_{u_m}\left[D^{\alpha_0}\left[\varphi \circ G_\theta\right]\right]
        & = \sum_{i=1}^I \sum_{j=1 }^{d^{\kappa_i}}  D_{u_m}\left[D^{\beta_{ij}} \varphi \circ G_\theta\right] \prod_{l=1}^{\kappa_i} \partial^{\alpha_{ijl}} G_{\theta, \eta_{ijl} } \nonumber \\
        &\qquad + \sum_{i=1}^I \sum_{j=1 }^{d^{\kappa_i}} \left[ D^{\beta_{ij}} \varphi \circ G_\theta \right] \partial^{e[m]} \Big[\prod_{l=1}^{\kappa_i} \partial^{\alpha_{ijl}} G_{\theta, \eta_{ijl} }\Big].
\label{eq:induction_proof_eq}
\end{talign}   
By the product rule for regular derivatives,
\begin{talign}   
    \partial^{e[m]} \Big[\prod_{l=1}^{\kappa_i} \partial^{\alpha_{ijl}} G_{\theta, \eta_{ijl} }\Big] = \sum_{l_0=1}^{\kappa_i} \partial^{\alpha_{ijl_0} + e[m]} G_{\theta, \eta_{ijl_0}} \prod_{\substack{l \in [1, \kappa_i] \\ l \neq l_0}} \partial^{\alpha_{ijl}} G_{\theta, \eta_{ijl} }.
\end{talign}   

Since $D^{\beta_{ij}} \varphi \in \W^{1, 2}(\X)$, the statement in~\Cref{lemma:chain_rule_w1p} applies to its composition with $G_\theta$, meaning
\begin{talign*}
    D_{u_m}\left[D^{\beta_{ij}} \varphi \circ G_\theta\right] = \sum_{j_0=1}^d \left[D_{x_{j_0}} \left[D^{\beta_{ij}} \varphi\right] \circ G_\theta\right] \frac{\partial {G_{\theta, j_0}}}{\partial u_m} = \sum_{j_0=1}^d \left[D^{\beta_{ij} + e[j_0]} \varphi \circ G_\theta\right] \frac{\partial {G_{\theta, j_0}}}{\partial u_m},
\end{talign*}
where, recall, $e[j_0]$ is a $d$-dimensional unit vector with $1$ as the $j_0$'th element. Substituting these into~\eqref{eq:induction_proof_eq}, we get
\begin{talign}
\begin{split}
\label{eq:induction_proof_eq2}
        D^{\tilde{\alpha}_0}[\varphi \circ G_\theta] 
        &= \sum_{i=1}^I \sum_{j=1 }^{d^{\kappa_i}}  \sum_{j_0=1}^d \left[D^{\beta_{ij} + e[j_0]} \varphi \circ G_\theta\right] \frac{\partial {G_{\theta, j_0}}}{\partial u_m} \prod_{l=1}^{\kappa_i} \partial^{\alpha_{ijl}} G_{\theta, \eta_{ijl} }  \\
        &\qquad \qquad \qquad \qquad \quad+\sum_{i=1}^I \sum_{l_0=1}^{\kappa_i} \sum_{j=1 }^{d^{\kappa_i}} \left[ D^{\beta_{ij}} \varphi \circ G_\theta \right] \partial^{\alpha_{ijl_0} + e[m]} G_{\theta, \eta_{ijl_0}} \prod_{\substack{l \in [1, \kappa_i] \\ l \neq l_0}} \partial^{\alpha_{ijl}} G_{\theta, \eta_{ijl} }
\end{split}
\end{talign}  
Now all that is left to do is find $\tilde{I}, \tilde{\kappa}, \tilde{\beta}, \tilde{\alpha}, \tilde{\eta}$ for which this will that the $(\tilde{\kappa}, \tilde{\beta}, \tilde{\alpha}, \tilde{\eta})$--form similar to~\Cref{eq:general_faa_di_bruno}. One can already see this should be possible, due to the flexibility in the definition of $(\tilde{\kappa}, \tilde{\beta}, \tilde{\alpha}, \tilde{\eta})$--forms; for completeness, we give the exact values now.

Define $\kappa_0 = 0$. Take $\tilde{I} = I+\sum_{i=1}^I \kappa_i$, $\tilde{\kappa}_i=\kappa_i+1$ for $i \in [1, I]$ and $\tilde{\kappa}_i=\kappa_p$ when $i \in [I+\sum_{j=0}^{p-1} \kappa_j, I+\sum_{j=0}^{p} \kappa_j]$ for $p \in [1, I]$, $p\in \mathbb{N}$, and
\begin{talign*}
    \tilde{\beta}_{ij} = 
    \begin{cases}
        \beta_{i\lfloor j/d \rfloor} + e[j \mod d], & i \in [1, I],\ j \in [1, 2^{\kappa_i+1}], \\
        \beta_{pj}, & i \in (I+\sum_{j=0}^{p-1} \kappa_j, I+\sum_{j=0}^{p} \kappa_j],\ j \in [1, 2^{\kappa_p}] \text{ for } p \in [1, I],
    \end{cases}
\end{talign*}
\begin{talign*}
    \tilde{\alpha}_{ijl} = 
    \begin{cases}
         \alpha_{i\lfloor j/d \rfloor l}, & i \in [1, I],\ j \in [1, 2^{\kappa_i+1}],\ l \in [1, \kappa_i], \\
         e[m], & i \in [1, I],\ j \in [1, 2^{\kappa_i+1}],\ l = \kappa_i+1, \\
         \alpha_{pjl}, & i \in (I+\sum_{j=0}^{p-1} \kappa_j, I+\sum_{j=0}^{p} \kappa_j],\ j \in [1, 2^{\kappa_p}],\ l \in [1, \kappa_p]\setminus \{i - I-\sum_{j=0}^{p-1} \kappa_j\}  \text{ for } p \in [1, I], \\
         \alpha_{pjl} + e[m], & i \in (I+\sum_{j=0}^{p-1} \kappa_j, I+\sum_{j=0}^{p} \kappa_j],\ j \in [1, 2^{\kappa_p}],\ l = i - I-\sum_{j=0}^{p-1} \kappa_j \text{ for } p \in [1, I],
    \end{cases}
\end{talign*}
\begin{talign*}
    \tilde{\eta}_{ijl} = 
    \begin{cases}
        \eta_{i\lfloor j/d \rfloor l}, & i \in [1, I],\ j \in [1, 2^{\kappa_i+1}],\ l \in [1, \kappa_i], \\
        j \mod d , & i \in [1, I],\ j \in [1, 2^{\kappa_i+1}],\ l = \kappa_i+1, \\
        \eta_{pjl}, & i \in (I+\sum_{j=0}^{p-1} \kappa_j, I+\sum_{j=0}^{p} \kappa_j],\ j \in [1, 2^{\kappa_p}],\ l \in [1, \kappa_p]  \text{ for } p \in [1, I].
    \end{cases}
\end{talign*}
where $j \mod d$ is the remainder of dividing $j$ by $d$. Then,~\eqref{eq:induction_proof_eq2} becomes
\begin{talign*}
    D^{\tilde{\alpha}_0}[\varphi \circ G_\theta] = \sum_{i=1}^{\tilde{I}} \sum_{j=1 }^{d^{\tilde{\kappa}_i}} \left[ D^{\tilde{\beta}_{ij}} \varphi \circ G_\theta \right] \prod_{l=1}^{\tilde{\kappa}_i} \partial^{\tilde{\alpha}_{ijl}} G_{\theta, \tilde{\eta}_{ijl} }.
\end{talign*}
This completes the proof of the induction step, and the theorem.
\end{proof}

\subsection{Proof of~\Cref{thm:rate_of_convergence}}

Before proving the main theorem, we introduce two auxilliary lemmas,~\Cref{lemma:rkhs_of_matern,lemma:bq_convergence}. The former will allow us to apply the chain rule of~\Cref{thm:k_circ_G_is_Sobolev} to get $k(x, \cdot) \circ G_\theta \in \Hc$, and the latter claim the asymptotic rate of $m^{-\nu_c/s - 1/2}$. The proof of~\Cref{thm:rate_of_convergence} will follow. 

Given~\Cref{as:points,as:generator,as:kernels}, all that is missing to prove $k(x, \cdot) \circ G_\theta \in \Hc$ by applying~\Cref{thm:k_circ_G_is_Sobolev} is the connection between RKHS of Mat\'ern kernels, and Sobolev spaces.
To that end, we introduce a Lemma (that is a minor extension to classic results, see for instance~\citet[Corollary 10.48]{Wendland2005}) that links RKHS $\Hk$ of a Mat\'ern kernel $k$ of order $\nu$ with Sobolev spaces $\W^{l, 2}$ for $l \in \mathbb{N}_+$. In the Lemma, we briefly refer to Bessel potential spaces---only for their norm-equivalence both to the RKHS of Mat\'ern kernels, and to the Sobolev spaces of fractional order, which themselves lie between Sobolev spaces of integer order---and to extension operators, that allow us to extend results on $\R^d$ to open, connected, bounded $\X$ with a Lipschitz boundary. Every open, bounded, and convex $\X$ has a Lipschitz boundary~\citep{stein1970singular}; for example, this includes the hypercube $\X=(0, 1)^d$. For a detailed overview of Bessel potential spaces, fractional Sobolev spaces, and extension operators, we refer to~\citet{adams2003sobolev}; these will only appear in the proof of the following Lemma.

For a $\beta \in \R_+ \cup \{0\}$, we denote the ceiling operation $\lceil \beta \rceil = \min(\{ z \in \mathbb{N} \ |\ z \geq \beta \})$, and the rounding operation $\lfloor \beta \rfloor = \max(\{ z \in \mathbb{N} \ |\  z \leq \beta \})$.

\begin{lemma}
\label{lemma:rkhs_of_matern}
    Suppose $\X=\R^d$, or $\X \subseteq \R^d$ is open, connected, and bounded with a Lipshitz boundary, and $k$ is a Mat\'ern kernel on $\X$ of order $\nu$. Then, the RKHS $\Hk$ induced by $k$ lies between Sobolev spaces $\W^{\lceil \nu + d/2 \rceil,2}(\X)$ and $\W^{\lfloor \nu + d/2 \rfloor,2}(\X)$, meaning
    \begin{talign*}
        \W^{\lceil \nu + d/2 \rceil,2}(\X) \subseteq \Hk \subseteq \W^{\lfloor \nu + d/2 \rfloor,2}(\X).
    \end{talign*}
\end{lemma}
\begin{proof}
    We start by proving the result for $\X=\R^d$. By~\citet[Corollary 10.13]{Wendland2005}, the RKHS of a Mat\'ern $k$ is norm-equivalent to the Bessel potential space $H^s(\X)$ for $s=\nu + d/2$. The Bessel potential space $H^s(\X)$, by~\citet[Section 7.62]{adams2003sobolev}, is norm-equivalent to a fractional Sobolev space (a Sobolev-Slobodeckij space) $W^{s,2}(\X)$, which lies between spaces of integer order, $W^{\lceil s \rceil,2}(\X) \subseteq W^{s,2}(\X) \subseteq W^{\lfloor s \rfloor,2}(\X)$. 
    
    Finally, the result $\W^{\lceil s \rceil,2}(\X) \subseteq \Hk \subseteq \W^{\lfloor s \rfloor,2}(\X)$ extends to an open connected bounded $\X\subset\R^d$ with a Lipshitz boundary identically to the proof of~\citet[Corollary 10.48]{Wendland2005}, which makes use of the extension operator introduced for such $\X$ by~\citet{stein1970singular}.
\end{proof}

To show the claimed asymptotic rate, we use the following straightforward corollary of~\citet[Theorem 9]{Wynne2020}.

\begin{lemma}[Corollary of Theorem 9 in~\citet{Wynne2020}]
\label{lemma:bq_convergence}
    Suppose for any $m \geq M \in \mathbb{N}_+$,
    \begin{itemize}
        \item $\U$ is a measure on a convex, open, and bounded $\calU \subset \R^s$ that has a density $f_\U: \calU \to [0, C_\U']$ for some $C_\U'>0$.
        \item $\{u_i\}_{i=1}^m$ are such that the fill distance $h_m =\mathcal O(m^{-1/s})$.
        \item $\{w_i\}_{i=1}^m$ are the optimal weights obtained based on the kernel $c_{\beta_m}$ and measure $\U$, parametrised by $\beta_m \in B$ for some parameter space $B$,
        \item for any $\beta \in B$, $c_{\beta}$ is a Mat\'ern kernel of order $\nu_c$; $\nu_c$ is independent of $\beta$.
    \end{itemize}
    Then, for some $C_0$ independent of $m$ and $f$, and any $f \in \Hc$ with $\|f\|_{\Hc} = 1$,
    \begin{talign*}
        \left| \int_\calU f (u)\U(\d u) - \sum_{i=1}^m w_i f(u_i) \right| \leq C_0 m^{-\nu_c/s - 1/2}.
    \end{talign*}
\end{lemma}
\begin{proof}
    The expression on the left hand side of~\citet[Theorem 9]{Wynne2020} is $|\int_\calU f (u)\U(\d u) - \sum_{i=1}^m w_i f(u_i)|$; the notation from their paper to this result maps as $\theta \to \beta$, $p \to f_\U$, $\X \to \calU$, $x \to u$, $\Theta \to B$, and the prior mean $\mu(\beta)=0$ for any $\beta \in B$. First, we show the assumptions in the Theorem hold.
    
    {Assumption 1 (Assumptions on the Domain):} An open, bounded, and convex $\calU$ satisfies the assumption, as discussed in~\citet{Wynne2020}.
    
    {Assumption 2 (Assumptions on the Kernel Parameters):} Since $c_\beta$ is a Mat\'ern kernel of order $\nu_c$, the smoothness constant of $c_\beta$ is $\nu_c + s/2$ regardless of the value of $\beta \in B$, meaning $\tau(\beta)=\tau_c^-=\tau_c^+=\nu_c + s/2>s/2$. Lastly, the norm equivalence constants of ~\citet[Equation 3]{Wynne2020} are the same for all $\beta$---since the respective RKHS and Sobolev spaces are---so the set of extreme values $B_m^*$ is finite and does not depend on $m$; we denote $B_c^*=B_m^*$, to highlight that $B_c^*$ only depends on the choice of kernel family $c$ and not $m$.
    
    {Assumption 3 (Assumptions on the Kernel Smoothness Range):} As discussed in Assumption 2, $\tau(\beta)=\nu_c + s/2$ for any $\beta \in B$, so the set in the statement of Assumption 3 has only one element.
    
    {Assumption 4 (Assumptions on the Target Function and Mean Function):} The target function $f$ is in $\Hc$, meaning $\tau_f = \tau_c^- = \tau_c^+ = \nu_c + s/2$. The mean function $\mu(\beta)$ was taken to be zero, so has zero norm.

    Lastly, take $h_0$ such that $h_1 \leq h_0$; as we assumed $h_m = \mathcal O(m^{-1/s})$, it holds that $h_0 \leq h_m$ for all $m \geq 1$. Therefore, all the assumptions are satisfied and~\citet[Theorem 9]{Wynne2020} applies; moreover, the bounding expression is $C_0 m^{-\alpha/s}$ for $\alpha=\nu_c + s/2$ and some $C_0$ independent of $m$ and $f$ since
    \begin{itemize}
        \item $h_m = \mathcal O(m^{-1/s})$, and as $\tau_f=\tau_c^-=\tau_c^+=\nu_c + s/2$ as discussed in the verification of assumptions, $h_m^{\max(\tau_f, \tau_c^-)} = O(m^{-\nu_c/s - 1/2})$,
        \item the rest of the multipliers do not depend on $m$ and $f$: $C$ depends only on $\calU$, $s$, $\tau_f = \nu_c + s/2$, and $B^*$; $\|f_\U \|_{\L^2(\calU)}$ is a constant and finite since $f_\U$ is bounded above; $\tau_f-\tau_c^+=0$ so rising to its power produces $1$; the norm $\|f\|_{\Hc} = 1$; for any $m \geq M$, $\mu(\beta_m) = 0$.
    \end{itemize}
    This completes the proof.
\end{proof}

Now we are ready to prove the main theorem.
\begin{proof}[Proof of Theorem~\ref{thm:rate_of_convergence}]
To show $k(x, \cdot) \circ G_\theta \in \Hc$ for all $x \in \X$, first note that~\Cref{lemma:rkhs_of_matern} applies for both $\calU$ and $\X$ that satisfy~\Cref{as:points}: trivially for $\R^d$, and for an open, convex, and bounded space since it has a Lipschitz boundary~\citep{stein1970singular}. Since by~\Cref{as:kernels}, $k$ is a Mat\'ern kernel of order $\nu_k$, it holds by~\Cref{lemma:rkhs_of_matern} that $k(x, \cdot) \in \W^{l_\varphi, 2}(\X)$ for $l_\varphi=\lfloor\nu_k+d/2\rfloor$ and any $x \in \X$. Then, by~\Cref{thm:k_circ_G_is_Sobolev}, $k(x, \cdot) \circ G_\theta \in \W^{\tilde{l}, 2}(\calU)$, for a $G_\theta$ that satisfies~\Cref{as:generator}, and $\tilde{l}=\min(l_\varphi, l)=\min(\lfloor\nu_k+d/2\rfloor, l)$. By~\Cref{as:kernels}, $\nu_c \leq \min(\lfloor \nu_k + d/2 \rfloor, l) - s/2=\tilde{l}- s/2$, and it holds that $\tilde{l} \geq \nu_c + s/2$. Since $\tilde{l}$ is an integer, this implies $\tilde{l} \geq \lceil \nu_c + s/2 \rceil$, and we have that $\W^{\tilde{l}, 2}(\calU) \subseteq \W^{ \lceil \nu_c + s/2 \rceil, 2}(\calU)$. 
Finally, as $c$ is a Mat\'ern kernel or order $\nu_c$, by~\Cref{lemma:rkhs_of_matern} it holds that $ \W^{\lceil\nu_c+s/2 \rceil, 2}(\calU) \subseteq \Hc$, and we arrive at $k(x, \cdot) \circ G_\theta \in \Hc$.

Since $k(x, \cdot) \circ G_\theta \in \Hc$ holds, we can use~\Cref{thm:optimal_weights} and state
    \begin{talign*}
        | \text{MMD}_k(\P_\theta,\Q^n) - \text{MMD}_k(\P_\theta^m,\Q^n)| \leq K \times \text{MMD}_c \left(\mathbb{U}, \sum_{i=1}^m w_i \delta_{u_i}\right).
    \end{talign*}
    By the reproducing property, it holds that $\sup_{\substack{f \in \Hc \\ \|f\|_{\Hc} = 1}} | \int_\calU f (u)\P(d u) -  \int_\calU f (u)\Q(d u) |$ is equal to $\text{MMD}_c(\P, \Q)$ for any two distributions $\P, \Q$ on $\calU$. Then,
    \begin{talign*}
        \text{MMD}_c(\mathbb{U}, \sum_{i=1}^m w_i \delta_{u_i}) = \sup_{\substack{f \in \Hc \\ \|f\|_{\Hc}=1}} \left| \int_\calU f (u)\U(d u) - \sum_{i=1}^m w_i f(u_i) \right|.
    \end{talign*}
    The expression under the supremum is bounded by~\Cref{lemma:bq_convergence} with $C_0 m^{-\nu_c/s - 1/2}$, for $C_0$ independent of $m$ and $f$. Therefore, $\text{MMD}_c(\mathbb{U}, \sum_{i=1}^m w_i \delta_{u_i}) \leq C_0 m^{-\nu_c/s - 1/2}$, and the result holds.
\end{proof}

Note that while the result was formulated for the special case of convex spaces, it applies more generally to any open, connected, bounded $\X \subset\R^d$, $\calU \subset \R^s$ with Lipschitz-continuous boundaries---with no changes to the proof. The applicability to $\X=\R^d$ remains unchanged; $\calU$, however, must remain bounded for~\Cref{thm:k_circ_G_is_Sobolev} to hold.

\subsection{Computational and sample complexity}\label{app:cost_error}

We derive the condition under which the OW estimator achieves better sample complexity than the V-statistic for the same order of computational cost, see Table~\ref{tab:cost_error} for the rates.

Suppose the cost for both V-statistic and OW is $\mathcal{O}(\tilde{m})$. Then, the sample complexity for the V-statistic can be written in terms of $\tilde m$ as $\mathcal{O}(\tilde{m}^{-1/4})$. Similarly, for the OW estimator, the sample complexity in terms of $\tilde m$ is $\mathcal{O}(\tilde{m}^{-(\nu_c + \frac{s}{2})/3s})$. 
The more accurate estimator is therefore the one whose error rate goes to zero quicker. Therefore, the OW estimator is more accurate than the V-statistic if
\begin{talign*}
   \frac{\nu_c}{s} > \frac{1}{4}, 
\end{talign*}
which for the common choice of $\nu_c=5/2$ implies $s< 10$.

\begin{table}[]
\centering
\caption{Computational and sample complexity rates of the V-statistic and the OW estimator with respect to $m$.}
\begin{tabular}{@{}lll@{}}
\toprule
            & \multicolumn{1}{c}{\textbf{Cost}} & \multicolumn{1}{c}{\textbf{Error}} \\ \midrule
V-statistic & $\mathcal{O}(m^2)$                & $\mathcal{O}(m^{-\frac{1}{2}})$            \\
OW          & $\mathcal{O}(m^3)$                & $\mathcal{O}(m^{- \frac{\nu_c}{s} - \frac{1}{2}})$    \\ \bottomrule
\end{tabular}
\label{tab:cost_error}
\end{table}

\subsection{Derivation of closed-form kernel embeddings} 
\label{app:derivation_aligns}

We have $z_i = \int_{\calU} c(u_i, v) \U(\text{d} v)$, where $c:\mathcal U \times \mathcal U \rightarrow \R$ is the SE kernel parameterised by the lengthscale $l>0$, i.e., $c(u,v) = \sqrt{2\pi}l \varphi(u; v, l^2)$, where $\varphi$ is the Gaussian pdf. For $s>1$, we can write the kernel as $c(u, v) = \prod_{j=1}^s c(u_j, v_j) $. We now derive closed-form kernel embeddings for $z_i$ for different choices of the base space $\calU$ and the distribution $\U$.

For $\calU=[0, 1]^s$, and $\U$ the uniform distribution, i.e., $u_i \sim \text{Uniform}([0,1]^s)$, we get
\begin{talign*}
    z_i &= \prod_{j=1}^s \int_{[0,1]} c(u_{ij}, v_j) \text{d} v_j =  \prod_{j=1}^s \sqrt{2\pi}l \left[ \varphi(1; u_{ij}, l^2) - \varphi(0; u_{ij}, l^2) \right],
\end{talign*}
where $\varphi$ is the Gaussian cdf and $u_{ij}$ is the $j^\textup{th}$ element of $u_i$. 

In the case of $\calU=\R^s$, and $\U$ being the Gaussian distribution such that $u_i \sim \mathcal{N}( \mu, \Sigma)$, where $ \mu = [\mu_1, \dots, \mu_s]^\top$ and $\Sigma$ is the $s$-dimensional diagonal matrix with entries $(\sigma^2_1, \dots, \sigma^2_s)$, the closed-form embedding for $z_i$ reads 
\begin{talign*}
    z_i = \prod_{j=1}^s \int_{-\infty}^\infty c(u_{ij}, v_j) \varphi(v_j; \mu_j, \sigma^2_j) \text{d} v_j 
    & = \prod_{j=1}^s \sqrt{2\pi}l \int_{-\infty}^\infty \varphi(v_j; u_{ij}, l^2) \varphi(v_j; \mu_j, \sigma^2_j) \text{d} v_j\\
        &= \prod_{j=1}^s \sqrt{\frac{l^2}{(l^2 + \sigma^2_j)}} \exp\left( \frac{- (u_{ij} - \mu_j)^2}{2(l^2 + \sigma^2_j)}\right).
\end{talign*}
For the special case of $\Sigma = \mathrm{diag}(\sigma^2, \dots, \sigma^2)$, the expression simplifies to
\begin{talign*}
    z_i = \left( \frac{l^2}{l^2 + \sigma^2} \right)^{s/2} \exp\left(- \frac{\Vert u_i - \mu \Vert^2}{2(l^2 + \sigma^2)} \right).
\end{talign*}

\section{Additional Experimental details} \label{app:expDetails}

True parameter values of the benchmark simulators in \Cref{sec:benchmark} is given in \Cref{app:trueParam}. \Cref{app:mvgk} and \Cref{app:compositetest} provide additional results and details regarding the experiments in \Cref{sec:mvgk}. Finally, the link to the source code of the wind farm simulator is in \Cref{app:windfarm}.

\subsection{Benchmark Simulators}\label{app:trueParam}
We now provide further details on the benchmark simulators. For drawing iid or RQMC points, we use the implementation from \texttt{SciPy} \cite{SciPy}. Below, we report the parameter value $\theta$ used to generate the results in Table~\ref{tab:MMDerror} for each model. We refer the reader to the respective reference in Table~\ref{tab:MMDerror} for a description of the model and their parameters.

\textbf{g-and-k distribution}: $(A, B, g, k) = (3,1,0.1,0.1)$

\textbf{Two moons}: $(\theta_1, \theta_2) = (0,0)$

\textbf{Bivariate Beta}: $(\theta_1, \theta_2, \theta_3, \theta_4, \theta_5) = (1,1,1,1,1)$

\textbf{Moving average (MA) 2}: $(\theta_1, \theta_2) = (0.6,0.2)$

\textbf{M/G/1 queue}: $(\theta_1, \theta_2, \theta_3)= (1,5,0.2)$

\textbf{Lotka-Volterra}: $(\theta_{11}, \theta_{12}, \theta_{13})= (5,0.025,6)$

\subsection{Multivariate g-and-k}\label{app:mvgk}

The performance of the V-statistic and our OW estimator as a function of $\theta_3$ parameter of the multivariate g-and-k distribution is shown in \Cref{fig:mvgk_varyingG} (left). The observed effect on the performance is similar to that of \Cref{fig:mvgk_c}, where the error in the OW estimator increases as we vary $\theta_3$. The degradation in performance is not as severe as when varying $\theta_4$, indicating that the smoothness of the multivariate g-and-k generator is not impacted by $\theta_3$ compared to $\theta_4$. Both the uniform and the Gaussian embedding achieves better performance than the V-statistic, whose performance remains unaffected by $\theta_3$.

\begin{figure}[t]
	 \centering
         \includegraphics[width=0.4\textwidth]{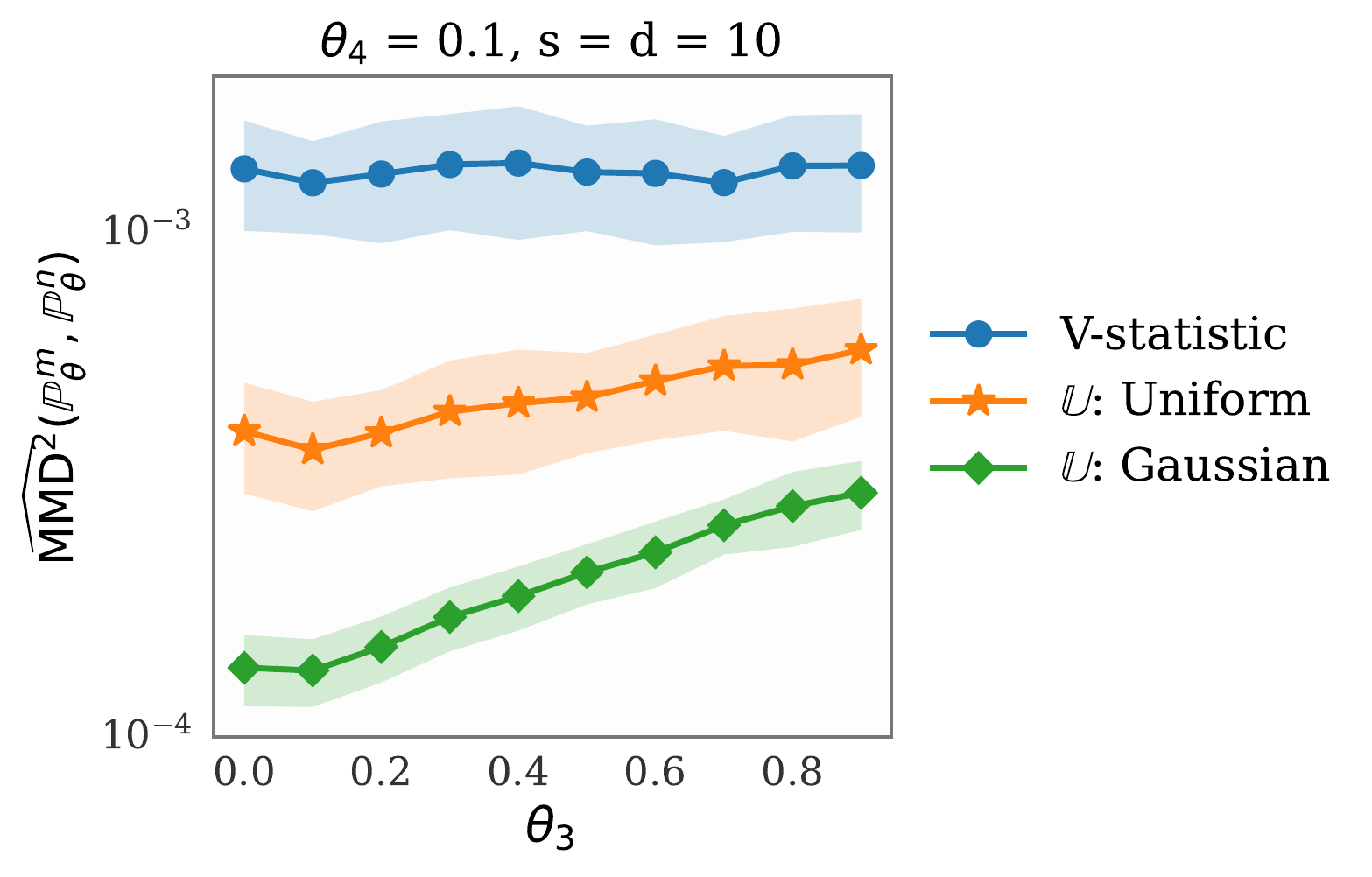}        
     \hspace{5mm}
         \includegraphics[width=0.37\textwidth]{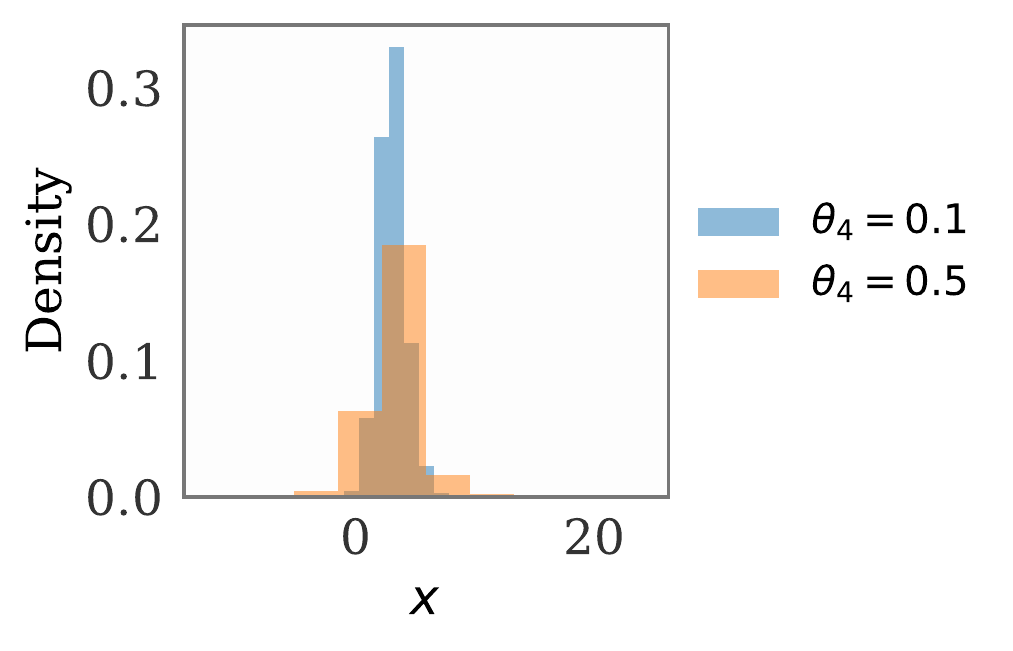}
\caption{Additional results for the multivariate g-and-k distributions. \textit{Left:} Estimated MMD$^2$ for the V-statistic and our OW estimator as a function of $\theta_3$. \textit{Right:} Histogram of samples from the g-and-k distribution for different values of $\theta_4$. Settings: no. of samples = 100,000, $\theta_1 = 3$, $\theta_2 = 1$, $\theta_3 = 0.1$.}
\label{fig:mvgk_varyingG}
\end{figure}

\subsection{Exchange rate data experiment}

We apply the g-and-k simulator to the US dollar to Canadian dollar exchange rate data \cite{Verbeek2018} from the \url{Ecdat} R package. The data is shown in \Cref{fig:exchange_data1} in black, which has $n=501$ points. We fit the univariate g-and-k model to this data using the ABC method of \Cref{eq:ABC}, with both the V-statistic and our OW estimator. For simplicity, we keep $\theta_3 = 0.12$ and $\theta_4 = 0.35$ fixed and only estimate the first two parameters. We set $m=20$ and simulate 2000 parameter values from the prior $\mathcal{U}([0.5,1] \times [0,0.1])$. The resulting ABC posteriors with tolerance $\varepsilon = 5\%$ are shown in \Cref{fig:exchange_posterior1} and \Cref{fig:exchange_posterior2} for $\theta_1$ and $\theta_2$, respectively. \Cref{fig:exchange_data1} shows the corresponding predictions based on the MAP estimate of the ABC posteriors. We observe that our OW estimator leads to a better fit than the V-statistic estimator. We are able to estimate the variance of the data (governed by $\theta_2$) much more accurately than the V-statistic, which overestimates the variance.

\begin{figure*}[t]
     \centering
     \begin{subfigure}[b]{0.245\textwidth}
         \centering
         \includegraphics[width=\textwidth]{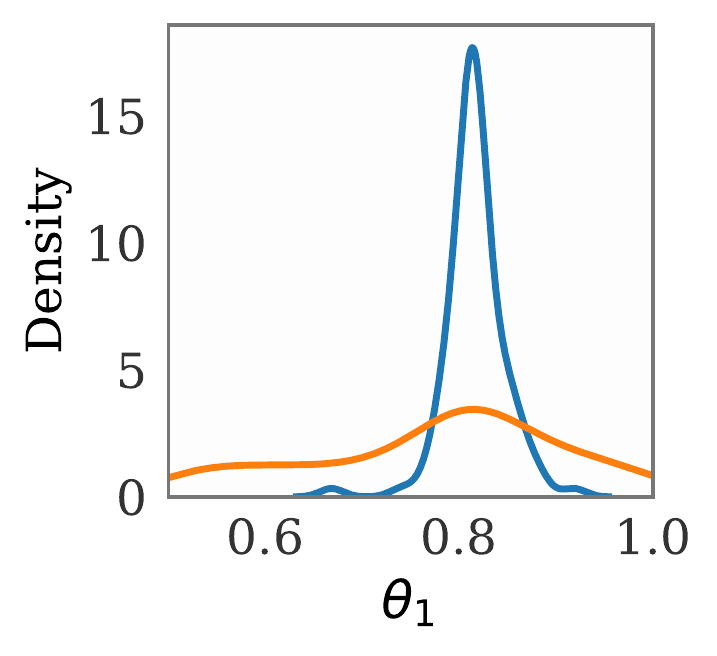}
         \caption{}
         \label{fig:exchange_posterior1}
     \end{subfigure}
     \begin{subfigure}[b]{0.25\textwidth}
         \centering
         \includegraphics[width=\textwidth]{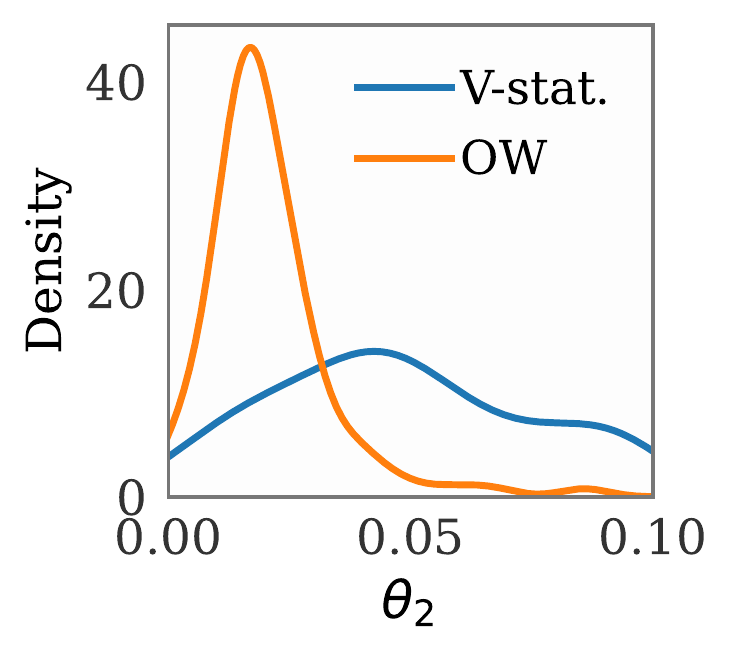}
         \caption{}
         \label{fig:exchange_posterior2}
     \end{subfigure}
     \begin{subfigure}[b]{0.235\textwidth}
         \centering
         \includegraphics[width=\textwidth]{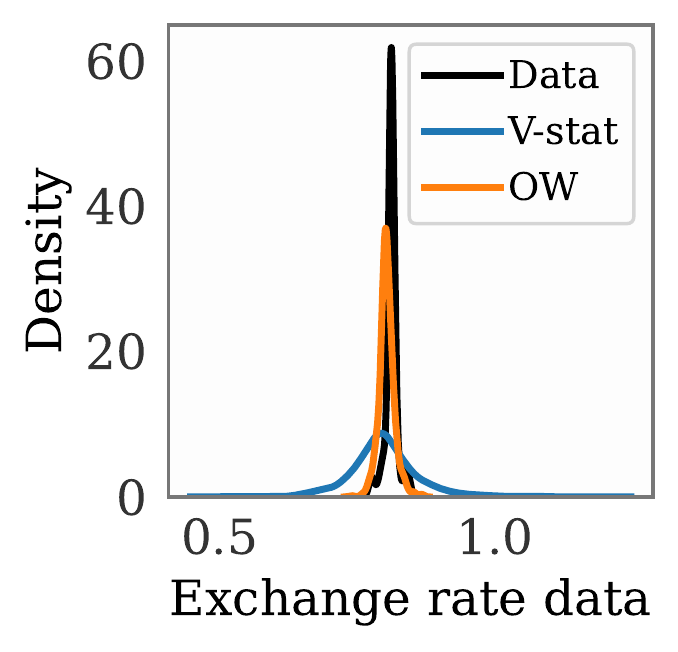}
         \caption{}
         \label{fig:exchange_data1}
     \end{subfigure}
     \begin{subfigure}[b]{0.255\textwidth}
         \centering
         \includegraphics[width=\textwidth]{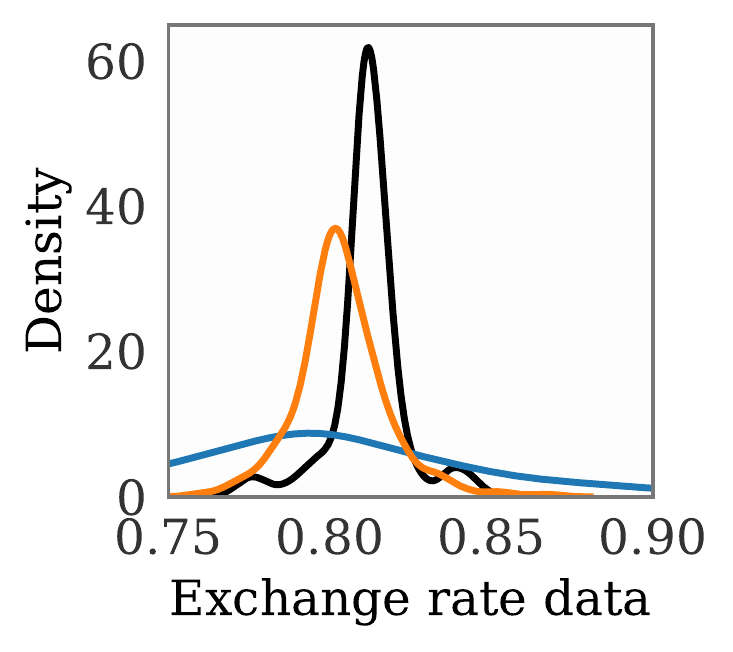}
         \caption{}
         \label{fig:exchange_data2}
     \end{subfigure}
        \caption{US dollar/Canadian dollar exchange rate data experiment on the g-and-k distribution. (a) ABC posterior for $\theta_1$. (b) ABC posterior for $\theta_2$. (c) Model fit based on the MAP estimates of $\theta_1$ and $\theta_2$ for V-statistic (in blue) and OW estimator (in orange). The OW estimator leads to a much better fit to the exchange rate data (shown in black). (d) Zoomed-in version of (c).}
        \label{fig:exchange}
        \vspace{-3mm}
\end{figure*}

\subsection{Composite goodness-of-fit test: details and additional results}
\label{app:compositetest}

\begin{figure}[h]
  \setlength{\algomargin}{0.2em}
  \begin{minipage}[t]{0.49\textwidth}
    \begin{algorithm2e}[H]
      \KwIn{$\P_{\theta}$, $\Q^n$, $\alpha$, $B$}
      $\hat{\theta}_{n} = \underset{\theta}{\arg \min} \operatorname{MMD}^2(\P_\theta, \Q^n)$ \;
      
      \For{$k \in \{1, \ldots, B\}$}{
          $\Q^n_{(k)} = \frac{1}{n} \sum_{i=1}^n \delta_{x_i^{(k)}}$, $\big\{x_i^{(k)}\big\}_{i=1}^n \sim \P_{\hat{\theta}_n}$\;
          
          $\hat{\theta}_{(k)}^n = \underset{\theta \in \Theta}{\arg \min} \operatorname{MMD}^2(\P_{\theta}, \Q^n_{(k)})$\;
          
          $\Delta_{(k)} = \operatorname{MMD}^2(\P_{\hat{\theta}_{(k)}^n}, \Q^n_{(k)})$\;
      }
      $c_\alpha = \operatorname{quantile}(\{\Delta_{(1)}, \ldots, \Delta_{(B)}\}, 1 - \alpha)$\;
      \BlankLine
      $\P_{\hat{\theta}_n}^m = \frac{1}{m} \sum_{i=1}^{m} \delta_{y_i}$, where $\{y_i\}_{i=1}^m \sim \P_{\hat{\theta}_n}$\;
      \eIf{$\operatorname{MMD}^2(\P_{\hat{\theta}_n}^m, \Q^n) > c_\alpha$}{\Return \text{reject}\;}{\Return do not reject\;}
    
      \caption{Composite goodness-of-fit test}
      \label{alg:composite_test}
    \end{algorithm2e}
  \end{minipage}
  \hfill
  \begin{minipage}[t]{0.49\textwidth}
    \begin{algorithm2e}[H]
      \KwIn{$\P_{\theta}$, $\Q^n$, $m$, $I$, $R$, $S$, $s$, $\Theta^\text{init}$}

      \SetKwProg{Fn}{Function}{ is}{end}
      \Fn{loss($\theta$)}{
        $\P_\theta^m = \frac{1}{m} \sum_{i=1}^m \delta_{y_i}$, where $\{y_i\}_{i=1}^m \sim \P_\theta$ \;
        \Return $\operatorname{MMD}^2(\P_\theta^m, \Q^n)$\;
      }
      
      $\theta^\text{trial}_{(1)},\ldots,\theta^\text{trial}_{(I)} \sim \Theta^\text{init}$ \;
      Select $\theta_{(1)}^\text{init}, \ldots, \theta_{(R)}^\text{init} \in \{\theta^\text{trial}_{(k)}\}_{k=1}^I$ that yield the smallest $\operatorname{loss}(\theta^\text{init}_{(k)})$\;
      $\hat{\theta}^\text{opt}_{(1)},\ldots,\hat{\theta}^\text{opt}_{(R)} = $
      \For{$k \in \{1, \ldots, R\}$}{
          $\hat{\theta}^\text{opt}_{(k)} = $ adam\_optimizer(loss, $S$, $s$, $\theta^\text{init}_{(k)}$)
      }
      \Return $\theta^* \in \{\hat{\theta}^\text{opt}_\text{(k)}\}_{k=1}^R$ s.t. $\forall k. \operatorname{loss}(\theta^*) \leq \operatorname{loss}(\hat{\theta}_{(k)}^\text{opt})$\;
    
      \caption{Random-restart optimiser}
      \label{alg:optimisation}
    \end{algorithm2e}
  \end{minipage}
\end{figure}

\Cref{alg:composite_test} shows the details of the composite goodness-of-fit test using the parametric bootstrap. The algorithm is written for the V-statistic estimator, but each instance of the squared MMD can be replaced with our OW estimator.
In practice, to compute ${\arg \min_\theta} \operatorname{MMD}^2(\P_\theta, \Q^n)$ we use gradient-based optimisation, as described in \Cref{alg:optimisation}.
The definitions of the hyperparameters of these two algorithms, and the values that we use, are given in \Cref{tab:hyper_def_app}.

\begin{table}[t]
\begin{center}
\centering
\caption{Definitions of the hyperparameters.}
\label{tab:hyper_def_app}
\begin{tabular}{c c l}
    \toprule
    hyperparameter & value & \\
    \midrule
    $\alpha$ & 0.05 & level of the test  \\
    $B$ & 200 & number of bootstrap samples \\
    $m$ & 100 & number of samples from the simulator \\
    $n$ & 500 & number of observations in the data \\
    $I$ & 50 & number of initial parameters sampled \\
    $R$ & 10 & number of initial parameters to optimise \\
    $S$ & 200 & number of gradient steps \\
    $s$ & 0.04 & step size \\
    \bottomrule
\end{tabular}
\end{center}
\end{table}

$\Theta_\text{init}$ is the distribution from which the initial parameters are sampled, and is a uniform distribution with the following ranges: $\theta_1: (0.001, 5)$, $\theta_2: (0.001, 5)$, $\theta_3: (0.001, 1)$, $\theta_5: (0.001, 1)$.
To compute the fraction of times that the null hypothesis is rejected (\Cref{tab:composite_test_results}) we repeat the experiment $150$ times.

\label{app:gof}

\subsection{Large scale wind farm model}\label{app:windfarm}

We include the comparison with the U-statistic estimator of MMD for the wind farm experiment in \Cref{fig:windfarm_Ustat}. Observations are similar to that of \Cref{fig:windfarm} --- our OW estimator leads to much more concentrated ABC posteriors around the true value than the U-statistic.

The low-order wake model is described in \citet{Kirby2023} and the code is available at \url{https://github.com/AndrewKirby2/ctstar_statistical_model/blob/main/low_order_wake_model.py}.

\begin{figure}[H]
	 \centering
	 \includegraphics[ width = 0.25 \linewidth]{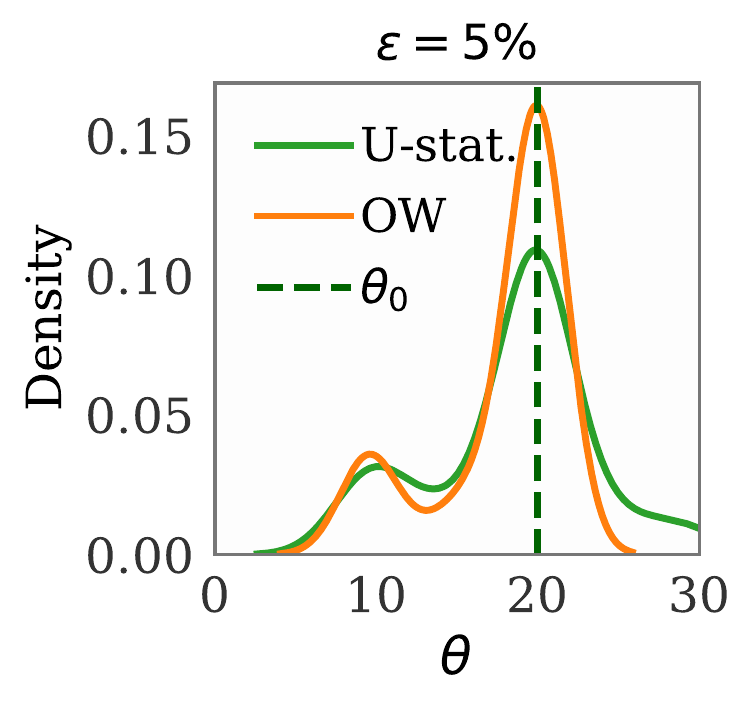}
	 \includegraphics[ width = 0.22 \linewidth]{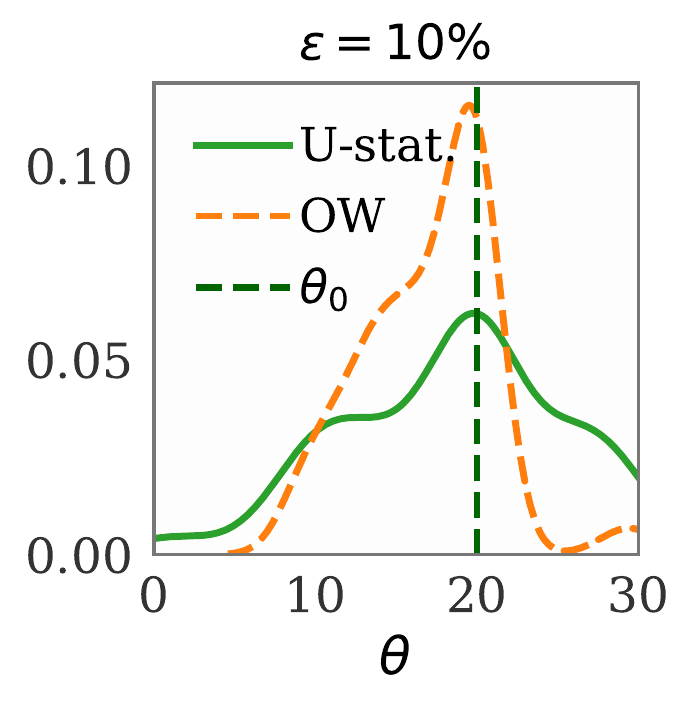}
\caption{ABC posteriors for the wind farm model. Our OW estimator yields posterior samples that are more concentrated around the true $\theta_0$ than the U-statistic. Settings: $n=100$, $\theta_0=20$.}
\label{fig:windfarm_Ustat}
\end{figure}

\end{document}